\documentclass[12pt]{amsart}
\usepackage[margin=1.25in]{geometry}
\usepackage[singlespacing]{setspace}
\usepackage{xfrac}

\usepackage[dvipsnames]{xcolor}
\usepackage[font=smaller]{caption}
\usepackage{tikz}
\usetikzlibrary{patterns} %

\pgfdeclarepatternformonly{wide vertical lines}{\pgfqpoint{-1pt}{-1pt}}{\pgfqpoint{10pt}{10pt}}{\pgfqpoint{2.5pt}{2.5pt}} %
{
  \pgfsetlinewidth{0.4pt}
  \pgfpathmoveto{\pgfqpoint{0pt}{0pt}}
  \pgfpathlineto{\pgfqpoint{0pt}{10pt}}
  \pgfusepath{stroke}
}
\pgfdeclarepatternformonly{wide horizontal lines}{\pgfqpoint{-1pt}{-1pt}}{\pgfqpoint{10pt}{10pt}}{\pgfqpoint{2.5pt}{2.5pt}} %
{
  \pgfsetlinewidth{0.4pt}
  \pgfpathmoveto{\pgfqpoint{0pt}{0pt}}
  \pgfpathlineto{\pgfqpoint{10pt}{0pt}}
  \pgfusepath{stroke}
}

\usepackage{amsthm}

\usepackage{amsmath,amsfonts, amssymb}

\usepackage{enumerate}
\usepackage{enumitem}
\usepackage{subfigure}
\usepackage{setspace}
\usepackage[ruled]{algorithm2e} %

\usepackage{etoolbox}
\usepackage{bbm}
\PassOptionsToPackage{hyphens}{url}

\usepackage{hyperref}

\hypersetup{
     colorlinks=true,
     linkcolor=blue,
     filecolor=blue,
     citecolor = blue,      
     urlcolor=cyan,
     }

\usepackage[authoryear,sort]{natbib}
 \bibpunct[, ]{(}{)}{,}{a}{}{,}%

\usepackage[normalem]{ulem}

\usepackage{graphicx}
\usepackage{color}

\newtoggle{version}
\toggletrue{version}

\newtoggle{comments}
\toggletrue{comments}

\theoremstyle{definition}
\newtheorem{example}{\sc Example}

\theoremstyle{definition}
\newtheorem{definition}{\sc Definition}

\newtheorem{lemma}{\sc Lemma}
\newtheorem{theorem}{\sc Theorem}
\newtheorem*{theorem*}{\sc Theorem}

\newtheorem{proposition}{\sc Proposition}

\iftoggle{comments}{
\newcommand{\at}[1]{{\color{orange} AT: #1}}
\newcommand{\sv}[1]{{\color{red}  SV: #1}}
\newcommand{\tn}[1]{{\color{purple} TN:  #1}}

\newcommand{\jh}[1]{{\color{cyan} JH: #1}}

}{
\newcommand{\at}[1]{}
\newcommand{\sv}[1]{}
\newcommand{\tn}[1]{}
\newcommand{\jh}[1]{}

}

\newcommand{\supply}[1]{c^{#1}}

\newcommand{\expect}[1]{\mathbb{E}\left[#1\right]}

\newcommand{\numtypes}{\tau}

\newcommand{\agent}{i} %

\newcommand{\good}{j}
\newcommand{\ngoods}{m} %
\newcommand{\agentset}{N}

\newcommand{\nagents}{n}
\newcommand{\goodsset}{M}

\newcommand{\typeset}{\Theta}

\newcommand{\samplespace}{\Omega}
\newcommand{\outcome}{\omega}

\newcommand{\optsub}[2]{{#1}_{#2}}

\newcommand{\lottery}[2]{\widetilde{{\mathbf #1}}_{#2}}

\newcommand{\prices}{{\bf p}}
\newcommand{\price}[1]{p^{#1}} %
\newcommand{\allocs}{\mathbf{X}}
\newcommand{\allocv}[1]{{\bf x}_{#1}}

\newcommand{\gallocv}[1]{{\bf x}_{#1}}

\newcommand{\bundle}{{\bf x}}

\newcommand{\ignore}[1]{}
\newcommand{\Xomit}[1]{}

\newcommand{\choices}{\Psi}
\newcommand{\budget}[1]{b_{#1}}

\newcommand{\mechanism}{\Phi}

\newcommand{\randombudget}[1]{\mathcal{B}_{#1}}

\newcommand{\optbundle}[1]{\optsub{\mathcal{X}}{#1}}

\newcommand{\z}{{\bf z}}
\newcommand{\y}{{\bf y}}
\newcommand{\x}{{\bf x}}
\newcommand{\e}{{\bf e}}
\newcommand{\ooo}{{\bf 0}}

\newcommand{\supplyvector}{\mathbf{\supply{}}}

\newcommand{\aceei}{ACEEI}

\newcommand{\gridjump}{\lambda}
\newcommand{\gridjumplong}{\numtypes \sqrt{n}}
\newcommand{\grid}[1]{\mathbf{G}(#1)}

\newcommand{\Mtwo}{CERI-SP}
\newcommand{\Mone}{CERI-L}
\newcommand{\CERIU}{CERI-S}
\newcommand{\nagtype}{\psi}
\newcommand{\nagtypevec}{\boldsymbol{\nagtype}}
\newcommand{\nagtypevecbig}{\boldsymbol{\Psi}}

\newcommand{\typedist}{\mathcal{D}}
\newcommand{\capacities}{\mathbf{c}}
\newcommand{\economy}{(\capacities, \succ)}
\newcommand{\zeconomy}{(\capacities, \typeapprox)}
\newcommand{\pseconomy}{(\capacities,\nagtypevec)}
\newcommand{\alteconomy}{(\capacities', \succ)}
\newcommand{\prob}{\mathbb{P}}
\newcommand{\support}{\text{supp}}
\newcommand{\expectation}{\mathbb{E}}

\newcommand{\minp}{p_{\min}(\typedist)}
\newcommand{\eventa}{A}
\newcommand{\eventb}{B}
\newcommand{\nthresh}{\bar{n}}
\newcommand{\typeapprox}{\overline{\nagtypevec}}
\newcommand{\typeapproxq}{\overline{\psi}}

\title[Efficiency, Envy, and Incentives in Combinatorial Assignment]{Efficiency, Envy, and Incentives \\ in Combinatorial Assignment}

\author{Th\`anh Nguyen \and Alexander Teytelboym  \and Shai Vardi}
\thanks{
Nguyen: Daniels School of Business, Purdue University; {\tt nguye161@purdue.edu}.
Teytelboym: Department of Economics, Institute for New Economic Thinking, and St.~Catherine's College, University of Oxford; {\tt alexander.teytelboym@economics.ox.ac.uk}.
Vardi: Muma College of Business, University of South Florida; {\tt vardi@usf.edu}.\\
An extended abstract of this paper appeared in the Proceedings of the 26th ACM Conference on Economics and Computation (EC'25). We are indebted to Aviad Rubinstein for the discussion that led to the conception of this paper. We are grateful to Eric Budish, Amrit Daswaney, Justin Hadad, Simon Jantschgi, Shengwu Li and Ludvig Sinander and many seminar participants at LBS, Monash, Columbia, Harvard and Tokyo for their terrific comments. 
This project has received funding from the European Research Council (ERC) under the European Union’s Horizon 2020 research and innovation programme (grant agreement No.~949699). This material is based upon work supported by the National Science Foundation under Grant No. DMS-1928930 and by the Alfred P. Sloan Foundation under grant G-2021-16778, while Teytelboym was in residence at the Simons Laufer Mathematical Sciences Institute (formerly MSRI) in Berkeley, California, during the Fall 2023 semester.
}

\begin{document}

\maketitle

\begin{abstract}

Ensuring efficiency and envy-freeness in allocating indivisible goods without money often requires randomization. However, existing combinatorial assignment mechanisms (for applications such as course allocation, food banks, and refugee resettlement) guarantee these properties either ex ante or ex post, but not both. We propose a new class of mechanisms based on Competitive Equilibrium from Random Incomes (CERI): Agents receive random token budgets and select optimal lotteries at competitive prices that clear markets in expectation. Our main insight is to let the CERI price vector guide all ex-post allocations. We show that all ordinally efficient allocations are CERI allocations, which can be implemented as lotteries over near-feasible Pareto-efficient outcomes. With identical budget distributions, CERI allocations are ordinally envy-free; with budget distributions on small supports, ex-post allocations are envy-free up to one good. Moreover, we design an asymptotically efficient implementation of CERI that satisfies a strong new non-manipulability property in large markets.

\end{abstract}

\onehalfspacing

\newpage
\section{Introduction}

Combinatorial assignment is a flexible and powerful model for allocation of indivisible resources without the use of money or of explicit priorities. 
Combinatorial assignment captures many complex allocation problems including course allocation \citep{budish2011combinatorial,budish2012multi}, allocation of food donations to food banks \citep{prendergast2017food,prendergast2022allocation}, and refugee resettlement \citep{ahani2021dynamic, delacretaz2023matching}.
In these settings, agents typically reveal their preferences over bundles of resources and the designer hopes to pick a  fair and efficient allocation.

In order to achieve fairness, mechanisms that allocate indivisible resources often use randomization.
It is therefore important to ensure that the mechanism's fairness and efficiency properties are desirable both ex ante (i.e., prior to the realization of the lotteries) and ex post. In this paper, we will focus on canonical ex-ante efficiency and fairness requirements for a random mechanism based only on ordinal preference information:  \emph{ordinal efficiency} and \emph{ordinal envy-freeness} \citep{bogomolnaia2001new}.\footnote{Perhaps better terms for these concepts would be ``ex-ante ordinal efficiency'' and ``ex-ante ordinal envy-freeness'' to contrast with ``ex-ante cardinal efficiency'' and ``ex-ante cardinal envy-freeness'' both of which fix a von Neumann-Morgenstern utility representation of preferences over lotteries. Note that ex-ante cardinal efficiency is stronger than ex-ante ordinal efficiency while ex-ante cardinal envy-freeness is weaker than ex-ante ordinal envy-freeness.}  Ordinal efficiency requires that no alternative assignment exist in which every agent receives a lottery that first-order stochastically dominates (ordered by their preferences) the original one.  Ordinal efficiency implies ex-post efficiency by ensuring that ex-post inefficient allocations are never assigned a positive probability. Ordinal envy-freeness requires that 
every agent receive a lottery that first-order stochastically dominates the lottery received by any other agent.
Since we only assume that agents' preferences over lotteries satisfy monotonicity (rather than be expected-utility maximizers)\footnote{Our reasons are practical rather than axiomatic. First, it is difficult to elicit preferences over lotteries over many consumption bundles. Second, in many market design applications it is preferable to make as few assumptions about agents' behaviour as possible to ensure that outcomes are robust to various forms of errors and misspecifications (this might be reminiscent of the ``Wilson's doctrine'' that criticizes the reliance on the common knowledge assumption in mechanism design). While the monotonicity assumption might appear weak, there is nevertheless some recent experimental evidence of its violation \citep{agranov2022revealed}.}, ordinal efficiency and ordinal envy-freeness are the natural ex-ante analogues of Pareto-efficiency and ex-post envy-freeness.\footnote{Of course, with indivisible goods, the set of ex-post envy-free outcomes can be empty even with two agents and one good, therefore, envy-freeness in the divisible-goods setting is the more reasonable analogue \citep{foley1967resource,varian1974equity}.}

There is no existing mechanism for combinatorial assignment that simultaneously achieves desirable ex-ante and ex-post fairness and efficiency properties.  Some ex-post efficient mechanisms that might appear fair, such as the Random Serial Dictatorship (RSD) and the Approximate Competitive Equilibrium from Equal Incomes (ACEEI) \citep{budish2011combinatorial}, are neither ordinally efficient nor ordinally envy-free. Other mechanisms, such as the Bundled Probabilistic Serial (BPS) \citep{nguyen2016assignment} are ordinally efficient and ordinally envy-free, but do not provide any ex-post envy-freeness guarantees. While the RSD mechanism is strategyproof, ACEEI and BPS mechanisms are only strategyproof in the large \citep{azevedo2019strategy}. Indeed, there is a stark tradeoff between ordinal efficiency, truthtelling incentives and minimal forms of fairness even for the simplest allocation problems \citep{bogomolnaia2001new}. 

The main insight of our paper is that by using a single market-clearing price vector to guide the ex-ante and ex-post allocations, the designer can simultaneously achieve ordinal efficiency, ordinal envy-freeness, approximate ex-post efficiency, approximate ex-post envy-freeness, as well as strong truthtelling incentives in large markets even in the most general combinatorial assignment settings.  
To achieve this, we introduce a new version of competitive equilibrium in an economy with an artificial currency (henceforth, ``tokens'') in spirit of \citet{varian1974equity} and \citet{budish2011combinatorial}. 
The key technical twist is that we make the budgets of tokens be \emph{exogenously random} for all agents. %
Our Competitive Equilibrium from Random Incomes (CERI) works as follows: agents are allocated distributions of token budgets, they report their ordinal preferences over bundles, and the designer computes their expected demand by averaging over optimal bundles at each budget realization.\footnote{We will use ``budget'' and ``income'' interchangeably.} 
Finally, the designer computes a linear and anonymous equilibrium price vector that equates aggregate expected demand to supply for every good. 
At CERI prices, an agent's allocation is therefore a lottery (or, more generally, a distribution) over optimal bundles at different budgets.
CERI prices turn out to be  precisely the prices that ensure desirable properties of the ex-ante and ex-post allocations.

First, we deal with existence and implementability.   
For a given profile of budget distributions, it might not be possible to exactly equate aggregate expected demand with supply (e.g., when budgets are identical and deterministic). 
However, we show that if distributions of budgets are continuous for all agents, then a CERI always exists (Theorem~\ref{theo:exist}).\footnote{The result is substantial since any discrete distribution of budgets can be made continuous by arbitrarily small perturbations.} 
Since CERI only outputs an ex-ante allocation (i.e., lottery for each agent), the key concern is whether a CERI allocation is implementable as a single lottery over ex-post feasible allocations. Unlike the unit-demand setting, we cannot merely rely on the Birkhoff-von Neumann (BvN) Theorem \citep{birkhoff1946tres,vonneumann1953assignment} to guarantee exact implementability. However, Theorem~\ref{theo:feasibility} shows that any CERI allocation can be implemented as a lottery over approximately feasible allocations supported by CERI prices and by appropriate realizations from correlated token budget distributions.%

Second, we turn to efficiency.
We show that any CERI allocation is ordinally efficient (i.e., our First Welfare Theorem).
We also show that any ordinally efficient allocation can be supported by a CERI with appropriate (independent) distributions of agents' token budgets (i.e., our Second Welfare Theorem).
Hence, an allocation is ordinally efficient if and only if it is a CERI allocation (Theorem~\ref{thm:characterization}). By connecting this characterization to approximate implementability established earlier, we show that any CERI allocation is implementable over near-feasible Pareto-efficient allocations (Theorem~\ref{theo:expostefficiency}).
While the economic content of our characterization---that efficient and equilibrium outcomes coincide in competitive markets---might be familiar from general equilibrium theory \citep{arrow1951extension,debreu1951coefficient}, the connection between ordinal efficiency and competitive pricing that we establish is novel. Additionally, we believe that our characterization of ordinally efficient allocations opens a door towards exploring market design applications in which token budgets are deliberately set unequally by the designer in order to correct for natural asymmetries between agents. For example, food banks are allocated different amounts of token currency to account for the different sizes of their client populations \citep{prendergast2017food}.

Third, we discuss envy-freeness. We show that if token budget distributions are identical then CERI produces an ordinally envy-free allocation (Theorem~\ref{theo:envy}). Moreover, if budget distributions are on a sufficiently small support, then any realization of CERI is ex-post envy-free up to one good (EF1) \citep{lipton2004approximately,budish2011combinatorial}. Therefore, we can simultaneously ensure that the lottery allocation is ordinally envy-free and that all ex-post allocations are EF1 by setting agents' budget distributions to be identical \emph{and} by requiring that they have a sufficiently small support.

The power of the single market-clearing CERI price vector becomes especially clear when we consider truthtelling incentives of mechanisms for combinatorial assignment.
As a warm-up, we combine the results above in Theorem~\ref{thm:ceris} to show that, under identical budget distributions, a mechanism (\CERIU) that uses CERI as its allocation rule is ordinally (and hence ex-ante cardinally) envy-free entailing that it is \emph{strategyproof in the large} \citep{azevedo2019strategy}. However, strategyproofness in the large provides a relatively weak incentive for truthtelling when we consider the whole market. Specifically, strategyproofness in the large and even stronger notions of asymptotic strategyproofness \citep{liu2016ordinal} do not ensure that, in any finite economy, a positive fraction of agents have a dominant strategy to truthfully report their preferences. To address this limitation, we introduce \emph{uniform strategyproofness} which is a far stronger incentive compatibility notion in large markets than those in the literature. Uniform strategyproofness requires that, as the market grows large, \emph{all agents} have a dominant strategy to truthfully report their preferences with a probability arbitrarily close to 1. To this end, we develop another CERI-based mechanism, called \Mone, which satisfies uniform strategyproofness while preserving asymptotic ordinal efficiency.

The technical idea behind \Mone\ is an exogenously \emph{random grid} of the type space. After agents report their types, the mechanism selects a grid point which approximates the type distribution. 
In \Mone, we compute a CERI using the grid point's  approximation of the type distribution and allocate bundles to agents using the prices from this CERI. 
The grid choice  balances efficiency and the provision of incentives. A coarse grid reduces the likelihood that an agent's misreport can affect the type distribution approximation, ensuring uniform strategyproofness. On the other hand, a fine grid closely reflects the original type distribution and leads to allocations that are close to efficient under truthful reporting. However, a finer grid gives the agents stronger incentives to misreport their preferences. 
By carefully choosing the size of the random grid, we demonstrate that \Mone~is uniformly strategyproof, asymptotically ordinally efficient, ordinally envy-free, and ex-post EF1 (Theorem~\ref{thm:m1}).
In ``small'' markets, away from the asymptotic limit, while \Mone~is not ordinally efficient, it can nevertheless be realized as a lottery over approximately feasible, Pareto-efficient allocations.

\begin{table}[ht]
\centering
\smaller

\begin{tabular}{c|ccccc}
\hline
    
    Mechanism/Property & Ordinally Eff. & Ex-post Eff. & Ordinally EF & Ex-post EF1 & Strategyproof \\
    \hline
    RSD  & No & Yes & No & No & Yes \\
    HBS Draft & No & Yes & No & No & No \\
    ACEEI & No & Yes$^*$ & No & Yes & In the large \\
    BPS  & Yes & Yes$^*$ & Yes & No & In the large \\
    \CERIU\ (this paper) & Yes & Yes$^*$ & Yes & Yes & In the large \\
    \Mone~(this paper)  & Asymp. & Asymp.$^*$ & Yes & Yes & Uniformly \\
    \hline
\end{tabular}
\caption{Our mechanism vs other combinatorial assignment mechanisms with ordinal preferences. {\textbf{Mechanisms:}} RSD = Random Serial Dictatorship. For the description of HBS Draft, see \citet{budish2012multi} and \citet{kominers2010course}. ACEEI = Approximate Competitive Equilibrium from Equal Incomes. BPS = Bundled Probabilistic Serial. \CERIU~= CERI mechanism with identical budget distributions over a small support (Theorem~\ref{thm:ceris}). \Mone~= CERI mechanism for a large market (Theorem~\ref{thm:m1}). {\textbf{Properties:}} EF = envy-free. EF1 = envy-free up to 1 good. Eff. = efficient. Yes$^*$ = the mechanism is $\Delta$-ex-post efficient (i.e., Pareto-efficient with respect to adjusted supply); see Definition~\ref{def:deltaeff}. Asymp. =  asymptotically; Asymp.$^*$ = asymptotically $\kappa$-ex-post efficient; see Section~\ref{sec:asymproperties}.}
    \label{tab:allocation}
    \end{table}

Taken together, our paper makes three key contributions. First, we introduce the idea of using a single market-clearing (CERI) price vector in order to achieve desirable ex-ante and ex-post properties of combinatorial assignment mechanisms (Table~\ref{tab:allocation}). Second, we offer a novel price-theoretic lens on existing mechanisms for assignment of indivisible resources. Third, we illustrate how CERI can be used to develop new mechanisms such as \Mone\, and to pinpoint the balance between efficiency, envy-freeness and incentives, both in large and small markets.

After reviewing related work, we present the combinatorial assignment model in Section~\ref{sec:model}. Then in Section~\ref{sec:ceri} we define CERI, provide  existence and implementability results, discuss the characterization of ordinally efficient allocations, and describe envy-freeness properties. In Section~\ref{sec:cerimech} we describe \CERIU\ and use CERI to offer a price-theoretic foundation for other assignment mechanisms in the literature. In Section~\ref{sec:strategyproof} we introduce \Mone\ and describe its (large-market) properties. Section~\ref{sec:conclusion} is a conclusion. 

\section{Relationship to existing work}
Our paper is related to three strands of the literature. First, we are able to illuminate the price-theoretic foundations of existing mechanisms (such as the Probabilistic Serial mechanism) for assignment with unit-demand agents. As a result, we can offer implementations of these mechanisms which have stronger incentive properties.
Second, we can combine the desirable properties of existing mechanisms for combinatorial assignment (ACEEI and Bundled Probabilistic Serial) into a single mechanism.
Third, our results are broadly related to a literature on pseudomarkets in which agents are assumed to be expected utility maximizers.
In summary, the benefit of using a single, anonymous price vector in CERI is that it fully characterizes ordinally efficient outcomes, allows us to obtain EF1 allocations and yields uniform strategyproofness. The cost is that in our setting ex-post efficient economies are only approximately feasible. While feasibility violations could be removed by using nonlinear pricing, it might add complexity, lose incentive guarantees and create potential opportunities for arbitrage in probability shares.

\subsection*{Ordinal preferences: unit-demand assignment}
Results from our paper shed a new light on existing assignment problems for unit-demand agents with ordinal preferences. In a celebrated paper, \citet{bogomolnaia2001new} introduced the Simultaneous Eating mechanisms that characterize all ordinally efficient allocations. A special case, the Probabilistic Serial mechanism in which all ``eating speeds'' are the same, delivers an ordinally envy-free outcome. Hence, for the unit-demand setting, CERI allocations (for different profiles of budget distributions) coincide with the allocations produced by Simultaneous Eating  mechanisms (for different profiles of ``eating speeds''). Our characterization entails that in the unit-demand setting an allocation is ordinally efficient if and only if every ex-post efficient allocation in its support can be supported by a single, anonymous equilibrium price vector. Moreover, an ordinally efficient allocation is ordinally envy-free if and only if agents face the same budget distributions in the CERI. In Section~\ref{sec:cerimech}, we provide a mapping between the ``eating speeds'' and ``finish times'' in the Probabilistic Serial mechanism and the budgets and prices in a CERI.

\citet{bogomolnaia2001new} showed that the Probabilistic Serial mechanism is not strategyproof.\footnote{Moreover, they showed that there are no ordinally efficient and strategyproof mechanisms that satisfy the equal treatment of equals property when the numbers of agent/good is greater than four.} However, \citet{kojima2010incentives} noted that the Probabilistic Serial mechanism is  strategyproof in a large enough replica economy. \citet{che2010asymptotic} demonstrated the asymptotic equivalence of Probabilistic Serial to the Random Serial Dictatorship. \citet{liu2016ordinal} then established a broader equivalence between asymptotically strategyproof and asymptotically efficient mechanisms. Compared to previous work on large markets, we strengthen the notion of asymptotic strategyproofness to uniform strategyproofness. Moreover, a special case of \Mone\ is an asymptotically ordinally efficient and uniformly strategyproof implementation of the Probabilistic Serial mechanism.

\subsection*{Ordinal preferences: combinatorial assignment}
Our work is intimately related to combinatorial assignment mechanisms with ordinal preferences \citep{budish2011combinatorial,budish2012multi,kornbluth2021undergraduate}.
The most celebrated such mechanism is the ACEEI introduced by \citet{budish2011combinatorial}. Budish shows that one can find an equilibrium price vector which exactly clears a convexified market after a small budget perturbation. He then shows that one can find a budget perturbation which ensures approximate ex-post market-clearing. Moreover, ACEEI allocations are ex-post EF1.

However, the ACEEI mechanism might be neither ordinally efficient nor ordinally envy-free. In the unit-demand setting, the ACEEI mechanism coincides with RSD (see Section~\ref{sec:cerimech}). For the multiunit-demand setting, we provide an example in Appendix~\ref{app:accei} that demonstrates that there may be only one deterministic allocation---an empty one---that satisfies ACEEI's ex-post market-clearing condition, but which is ordinally inefficient. The reason that the ACEEI mechanism is neither  ordinally efficient nor ordinally envy-free is that it computes a different set of market-clearing prices for different budget perturbations, rather than, as CERI does, a single set of market-clearing prices for a given profile of budget distributions.\footnote{Our construction of a CERI is similar to \citeauthor{budish2011combinatorial}'s construction of his convexified equilibrium, however, by using a different rounding procedure we can avoid \citeauthor{budish2011combinatorial}'s initial budget perturbation step that compromises ordinal efficiency and ordinal envy-freeness of the ACEEI mechanism.} As a consequence, CERI allocations are ordinally efficient (and ordinally envy-free with appropriate budget distributions) and, in fact, we can show that the existence of a CERI implies the existence of an ACEEI (Appendix~\ref{app:B}).
A final difference between the properties of ACEEI and of CERI is that the ACEEI allocation bounds depend on the market size and are expressed in terms of the $\ell^2$-norm whereas CERI bounds are expressed good-by-good and are independent of the market size (i.e., in the $\ell^\infty$-norm). The good-by-good bounds which are invariant to the market size are often more practical from a market design perspective because by adjusting the capacities of individual goods the designer can ensure that none of the bounds are exceeded ex-post  \citep[p. 4136-4137]{nguyen2021delta}.
Hence, our \CERIU\ mechanism replicates the approximate ex-post efficiency, ex-post EF1, and strategyproofness in the large properties of ACEEI while additionally ensuring ordinal efficiency and ordinal envy-freeness (Theorem~\ref{thm:ceris}).

\citet{nguyen2016assignment} proposed the BPS mechanism for combinatorial assignment. The BPS mechanism extends the PS mechanism to the combinatorial setting. The BPS mechanism is ordinally efficient and ordinally envy-free. In Section~\ref{sec:cerimech}, we provide a mapping between the ``eating speeds'' and ``finish times'' in the BPS mechanism and the budgets and prices in a CERI. But, unlike CERI, BPS neither guarantees ex-post EF1 nor characterize all ordinally efficient allocations (as shown in Example~\ref{ex:eating}).

Overall, mechanisms based on CERI can combine the best features of ACEEI---i.e., ex-post efficiency after adjusting supply and EF1---with the best features of the BPS mechanism---ordinal efficiency and ordinal envy-freeness---while strengthening the incentive properties of these mechanisms in large markets \citep{azevedo2019strategy} (see Table~\ref{tab:allocation}).

\subsection*{Cardinal preference representations}
There is a substantial literature on efficient and fair allocation of indivisible goods in which agents are assumed to be expected utility maximizers.\footnote{There is a well known connection between models with ordinal and vNM preferences. \citet{mclennan2002ordinal} showed that any ordinally efficient lottery allocation maximizes the sum of expected utilities for some vector of vNM utility functions that are consistent with the ordinal preferences. See also \citet{manea2008constructive} and \citet{carroll2010efficiency}. }
\citet{aziz2023best} explored how to combine various desirable ex-ante and ex-post properties (what they called the ``best of both worlds'') and derived a number of impossibility results under the assumption of additive preferences over bundles.  CERI mechanisms can also claim to satisfy the ``best of both worlds'' property.
\citet{cole2021existence} show that for a class of partition-based utility functions one can guarantee allocations that are exactly realizable (i.e., the lottery allocation is assumed to come from the set of lotteries over ex-post feasible allocations), ex-ante cardinally efficient and ex-ante cardinally envy-free. Their notion of efficiency is stronger than ours, but their notion of envy-freeness is weaker. Moreover, imposing exact implementability means that ex-post efficient allocations are supported by different market-clearing price vectors which weakens the large-market incentive compatibility properties of their solution.
Finally, a strand of work explored ``pseudomarket'' allocation rules (reminiscent of ours) in which expected-utility maximizing agents can select their most preferred lotteries at competitive prices using budgets of tokens \citep{hylland1979efficient,budish2013designing,gul2019market,nguyen2021delta, echenique2021constrained}. In this setting, \citet{miralles2021foundations} provide First and Second Welfare Theorems while \cite{nguyen2024equilibrium} provide necessary and sufficient conditions for exact implementability in a sufficiently rich preference domain.

We only assume that agents' preferences are monotonic in probabilities so the designer is only required to elicit an ordinal ranking over bundles. Indeed, there is evidence that relying on ordinal preferences is more robust from a practical market design perspective: Recent work suggests that market participants (e.g., in course allocation) make more mistakes when their ranking over bundles is sensitive to the cardinal representation of preferences \citep{budish2022can}.

\section{Model}\label{sec:model}
There is a finite set  $\goodsset$ of \emph{goods}, with $|\goodsset|=\ngoods$,  and a finite set $\agentset$ of agents, with $|\agentset| = \nagents$. Each good $\good$ has a finite integer \emph{capacity} $\supply{\good}\in \mathbb{N}$.

A \emph{bundle} is an integral vector, $\bundle \in \mathbb{N}^m$.  
If $\x$ contains good $\good$, we use $(\x-\e^j)^+$ to denote the bundle $\x$ with one unit of good $\good$ removed.   
The bundle consumed by agent $\agent$ is denoted by~$\bundle_{\agent}$. There might be further constraints on consumption; let $\choices_\agent \subseteq  \mathbb{N}^m $ denote the set of \emph{acceptable} bundles for agent $\agent$. We assume that $\ooo \in \choices_\agent$ for all $\agent$, meaning each agent has an outside option. However, we do not assume  free disposal, indicating that if a bundle $\bundle_{\agent}$ is acceptable to agent $i$, it is possible that $\bundle'_{\agent}\leq \bundle_{\agent}$ is not acceptable. Furthermore, we assume the maximum size of an acceptable bundle for any agent is at most $\Delta$.\footnote{Formally, $\Delta:=\max _{i \in N} \max _{x \in \Psi_i} \sum_{j \in M} x_j$. In applications, such as course allocation or refugee resettlement, $\Delta$ is around 6.}
An (ex-post or deterministic) \emph{allocation} $\allocs=(\gallocv{1}, \ldots, \gallocv{\nagents})$ is a list of acceptable bundles, one for each agent. Allocation $\allocs$ is \emph{feasible} with respect to capacities $\supplyvector$ if  $\sum_{\agent \in \agentset} \allocv{\agent} \leq \supplyvector$.  

Each agent $\agent$ has a strict preference relation $\succ_\agent$ over the set $\choices_\agent$. We assume {(without loss)} 
 that $\ooo$ is the least preferred bundle for all agents. 
We denote the weak relation of $\succ_\agent$ by $\succeq_\agent$; i.e.,  $\x\succeq_\agent \y$ means  either $\x\succ_\agent \y$ or $\x=\y$, and denote the preference profile of  all agents by $\succ:= (\succ_\agent)_{\agent \in \agentset}$.  We  use  $\succ_{-i}$ to denote the preference profile of agents excluding agent $i$.

In addition to (ex-post) allocations, our paper considers (ex-ante) lottery allocations and their associated stochastic order. %
Let $\mathcal{L}(\choices_\agent)$ denote the set of lotteries over $\choices_\agent$. We use $\lottery{x}{i}\in \mathcal{L}(\choices_\agent)$ to indicate a lottery obtained by agent $i$ and  $\expect{\lottery{x}{i}}$ to denote expectation of this lottery.

A \emph{lottery allocation}, $\lottery{X}{}=(\lottery{x}{1},..,\lottery{x}{n})\in 
\mathcal{L}(\choices_1)\times ..\times \mathcal{L}(\choices_n)$, is a list of lotteries over the acceptable bundles, one for each agent.  
The lottery allocation $(\lottery{x}{1},..,\lottery{x}{n})$ is \emph{feasible}  with respect to capacity $\bf c$ if $\sum_{i=1}^n \expect{\lottery{x}{i}} \le  \bf c$.
A lottery allocation $\lottery{X}{}$ %
is \emph{implementable} over a set of ex-post allocations (which are not necessarily feasible) if it can be realized as a lottery over this set of allocations.
Since our primitives are ordinal preferences rather than their cardinal representations, we will assume that agent $i$  prefers lottery $\lottery{x}{}$ to lottery $\lottery{y}{}$ if and only if $\lottery{x}{}$ (first-order) stochastically dominates $\lottery{y}{}$ (when ordered by $\succ_i$).

\begin{definition} [\citeauthor{bogomolnaia2001new}, \citeyear{bogomolnaia2001new}]
For agent $i$, consider two lotteries $\lottery{x}{}, \lottery{y}{} \in \mathcal{L}(\choices_i)$. 
We say that $\lottery{x}{}$ \emph{stochastically dominates} $\lottery{y}{}$, denoted 
$\lottery{x}{}\succeq^{sd}_i \lottery{y}{}$, if for every bundle $\z \in \choices_i$, 
\[
\sum_{\x \succeq_i \z} \mathbb{P}_{\x}(\lottery{x}{}) \;\;\ge\;\; 
\sum_{\y \succeq_i \z} \mathbb{P}_{\y}(\lottery{y}{}),
\]
where $\mathbb{P}_{\mathbf{x}}(\lottery{x}{})$ denotes the probability that lottery $\lottery{x}{}$ 
yields outcome $\mathbf{x}$. We say that $\lottery{x}{}$ \emph{strictly stochastically dominates} 
$\lottery{y}{}$, denoted $\lottery{x}{}\succ^{sd}_i \lottery{y}{}$, if the inequality is strict for at least one bundle $\z$.
\end{definition}

Note that stochastic dominance places only very weak assumptions on agents' preferences over lotteries viz. that preferences satisfy monotonicity.\footnote{Most decision-theoretic models assume monotonicity although some recent models of preferences for randomization allow for violations of monotonicity  \citep{agranov2022revealed}.} In particular, we do not assume that agents are expected utility maximizers.\footnote{In the special case when agents are expected utility maximizers, the condition says that agent $i$ prefers $\lottery{x}{}$ to $\lottery{y}{}$ if the agent obtains higher expected utility from $\lottery{x}{}$ than from $\lottery{y}{}$ for \emph{any} von Neumann-Morgenstern (cardinal) representation of his ordinal preferences  $\succ_i$.}

An \emph{economy} is a tuple $\mathcal{E}=(N,M,\supplyvector{}, (\choices_\agent)_{\agent\in\agentset}, \succ)$. We shorten this to $\mathcal{E}=\economy$, as capacities $\mathbf{c}$ and preferences $\succ$ completely define the economy.
A (direct) \emph{mechanism} $\mechanism(\cdot)$ maps every economy to a lottery allocation, and  $\Phi_i(\cdot)$ denotes the lottery obtained by agent $i$ in this mechanism. {We say that a mechanism has an ex-ante property \textsf{P} if its lottery allocation has property \textsf{P}; we say that a mechanism has an ex-post property \textsf{Q} if \emph{all} realizations of the mechanism's lottery allocation have property \textsf{Q}.}

\section{Competitive Equilibrium from Random Incomes}
\label{sec:ceri}
We first describe our novel equilibrium concept.
Consider a random variable $\randombudget{}\ge 0$ which we call a \emph{random income}.  For any agent $\agent$,  price vector $\prices \ge 0$ and random income $\randombudget{\agent}$, define the following random variable
\begin{equation}\label{eq:xhat}
\optbundle{\agent}(\prices,\randombudget{\agent}) := \left\{\max_{\succ_{\agent}} \{\gallocv{}:  \gallocv{} \in \choices_{\agent} \text{ and } \prices\cdot \gallocv{} \leq \budget{\agent} \} \;\; \middle| \;\;  \budget{\agent}\sim \randombudget{\agent} \right\} .
\end{equation}
 We call $\optbundle{\agent}(\prices,\randombudget{\agent})$ the \emph{random demand} of the  agent $\agent$. The realizations of $\optbundle{\agent}$ are the   optimal bundles for agent $\agent$ at prices $\prices$ when $\agent$'s budget is drawn from the distribution $\randombudget{\agent}$. 
Denote agent $i$'s \emph{expected demand} by $\expect{\optbundle{\agent}(\prices,\randombudget{\agent})}$, where the expectation is over $\randombudget{\agent}$.  %

Using this definition, our equilibrium concept is intuitive: when each agent receives what they demand given their random incomes, the markets for all goods clear exactly in expectation.

\begin{definition}Given  an economy $\mathcal{E}=\economy$  and a profile of random incomes, \label{def:ceri} $(\randombudget{1},\ldots,\randombudget{\nagents})$,  the prices $\prices = (\price{1},\ldots,\price{\ngoods})$ and the lottery allocation $ \lottery{X}{}=(\lottery{x}{1},..,\lottery{x}{n})$ comprise a \emph{competitive equilibrium from random incomes} (CERI) if %
 
\begin{enumerate}%
\item[(i)]  $\lottery{x}{i}=\mathcal{X}_i(\prices,\randombudget{i})$ for each $i$, and
\item[(ii)] $\sum_{\agent \in \agentset} \expect{\lottery{x}{i}}_\good \le \supply{\good}$ for every good $\good$, with equality whenever $\price{\good}>0$.
 \end{enumerate}  
\end{definition}
We refer to a lottery allocation $\lottery{X}{}$ as a \emph{CERI allocation} if there exists a price vector such that, together with $\lottery{X}{}$, it forms a CERI. Similarly, a price vector is called \emph{CERI prices} if it corresponds to the price vector of a CERI. In Section~\ref{sec:cerimech}, we define a \emph{CERI mechanism}, which explicitly describes how the concept of CERI can be used in an allocation mechanism.

A CERI allocation is a profile of independent lotteries over acceptable bundles.
In order to make such an allocation economically meaningful, we must ensure that this lottery allocation can be implemented as a lottery over ex-post allocations.

\begin{definition}\label{def:implementation}
Given  an economy $\mathcal{E}=\economy$  and a profile of random incomes,  $(\randombudget{1},\ldots,\randombudget{\nagents})$, and a CERI $(\prices,\lottery{X}{})$, a lottery $\{\mathbf{X}^{\mathbf{w}} \; | \;\mathbf{w} =(w_1,..,w_n) \sim  \mathcal{W} \}$ over allocations is an \emph{ex-post implementation} of the CERI if
\begin{enumerate}
\item $\mathcal{W}$ is a joint income distribution over $\mathbb{R}^n$ for which each marginal distribution in coordinate~$i$ coincides with $\randombudget{i}$, and
\item in each supporting allocation $\mathbf{X}^{\mathbf{w}}:= (x^\mathbf{w}_1,..,x^\mathbf{w}_n)$ we have that $
x^\mathbf{w}_i:= \max_{\succ_{\agent}} \{\gallocv{}:  \gallocv{} \in \choices_{\agent} \text{ and } \prices\cdot \gallocv{} \leq w_i \}$ for all $i$. 
\end{enumerate}

\end{definition}

An ex-post implementation of a CERI is a joint distribution over incomes (in which the marginals coincide with the agents' budget distributions) supported by a set of ex-post allocations in which every agent gets their most preferred bundle under CERI prices.  This definition entails one of the key conceptual contributions of the paper: the allocations in the support of any ex-post implementation are governed by a single CERI price vector.

We call an ex-post implementation of a CERI \emph{feasible} if all supporting allocations are feasible, i.e., $\sum_{i=1}^n x_i^\mathbf{w} \leq \supplyvector$; in this case, each $\mathbf{X}^{\mathbf{w}}$ is a competitive equilibrium allocation with deterministic incomes $\mathbf{w}$ and supported by the CERI price vector $\prices$. When agents have unit demand, a feasible ex-post implementation of the CERI always exists by the Birkhoff--von Neumann Theorem \citep{birkhoff1946tres,vonneumann1953assignment}.
 Unfortunately, due to the presence of combinatorial demand, in general it is impossible to implement a lottery allocation over feasible ex-post allocations (no matter which prices they are supported by). The reason is intuitive: the joint distribution of incomes $\mathcal{D}$ might allow several agents to independently draw a large budget and demand the same desirable bundle that cannot be provided by the designer.
We therefore allow the designer to relax the capacity of any good in any ex-post allocation by a small amount in order to achieve approximate implementability.

\begin{definition}\label{def:nearimplementation}
Given  an economy $\mathcal{E}=\economy$  and a profile of random incomes,  $(\randombudget{1},\ldots,\randombudget{\nagents})$, and a CERI $(\prices,\lottery{X}{})$, we say that an ex-post implementation $\mbox{$ \{\mathbf{X}^\mathbf{w} \;|\; \mathbf{w}\sim\mathcal{W} \}$}$ is \emph{$\kappa$-near-feasible} if all supporting allocations $ (x^\mathbf{w}_1,..,x^\mathbf{w}_n)$ satisfy
$$
\sum_{i=1}^n x^\mathbf{w}_i \le  \supplyvector +\kappa\cdot {\bf 1}.\footnote{Note that the definition of $\kappa$-near-feasible allocations only relaxes the upper bound, but does not impose a lower bound. It is straightforward to adapt this definition and the results to include a lower bound. See footnote~\ref{fn:lowerbound}.}
$$

\end{definition}
To illustrate the definitions of CERI, exact ex-post implementation and $\kappa$-near-feasible ex-post implementation, consider the following example.
\begin{example}
Consider two unit-demand agents $\{1,2\}$ and two goods $\{a,b\}$. Both agents have the same preference ordering: $\{a\} \succ_{1,2} \{b\}$. Each agent's random income is either $(\$1\, \text{w.p.}\, \frac{1}{2}, \$2\, \text{w.p.}\, \frac{1}{2})$.  A CERI in this economy has prices $(p_a,p_b)=(2,1)$. Each agent receives each good w.p. $\frac{1}{2}$.  There are both feasible and near-feasible implementations of this CERI. One ex-post implementation occurs when the joint income distribution is $(\$1, \$2)$ w.p. $\frac{1}{2}$ and $(\$2, \$1)$  w.p. $\frac{1}{2}$, leading to ex-post allocations $(b, a)$ and $(a, b)$, respectively (where, e.g., $a$ denotes a unit-vector in coordinate of good $a$.).

Another joint income distribution is $(\$1, \$1)$ w.p. $\frac{1}{2}$ and $(\$2, \$2)$ w.p. $\frac{1}{2}$, yielding ex-post allocations $(b, b)$ and $(a, a)$. This is not an exactly feasible ex-post implementation of the CERI, but it is a 1-near-feasible ex-post implementation. 
\end{example}

\subsection*{Existence and Implementability}
A CERI is not guaranteed to exist for all profiles of budget distributions (e.g., when budgets are deterministic). Nonetheless, we establish that a CERI always exists when the budget distributions are continuous.

\begin{theorem}[Existence]\label{theo:exist}
Given  an economy $\mathcal{E}=\economy$  and a profile of random incomes,  $(\randombudget{1},\ldots,\randombudget{\nagents})$,  if $\randombudget{i}$  %
is continuous on its domain for all $i \in \agentset$, %
then a CERI exists.
\end{theorem} 

The crucial observation in the proof of Theorem~\ref{theo:exist} is that with continuous income distributions, random demand is continuous in the lottery space. This allows us to apply a standard fixed-point argument.

For discrete distributions of agents' budgets, however, a CERI might not exist. Yet, we can slightly perturb the budget profile so that a CERI exists. Specifically, for any discrete budget distribution $\mathcal{B}_i$, we can replace each mass $\rho$ with a uniform distribution on $[\rho, \rho + \epsilon]$ for any $\epsilon > 0$. The resulting profile of budget distributions will have cumulative distribution functions that are continuous, and therefore, a CERI will exist under the new profile.

Next, we discuss ex-post implementation of CERI. In general, there is no exact ex-post implementation of a CERI allocation. %
However, if the designer can relax the supply of any good by one less than the size of the largest bundle, any CERI allocation be near-feasibly implemented.   

\begin{theorem}[Implementability]\label{theo:feasibility}
    Fix an economy $\mathcal{E} = \economy$, and a profile of random incomes,  $(\randombudget{1},\ldots,\randombudget{\nagents})$. 
    If the maximum size of an acceptable bundle is $\Delta$, then any CERI has a $\Delta-1$-near feasible ex-post implementation.\footnote{\label{fn:lowerbound} Using methods from \cite{nguyen2021delta}, it is straightforward to show that if a symmetric lower bound were introduced in Definition~\ref{def:nearimplementation} of $\kappa$-near-feasibility then the feasibility error on both bounds of the ex-post implementation would double.}
\end{theorem}
The proof of Theorem~\ref{theo:feasibility} applies a generalization of the Birkhoff--von Neumann theorem due to \citet{nguyen2016assignment} and works as follows.
For each agent $i$, given any realization of the random demand, we can construct a conditional income distribution under which the agent’s consumption matches the realization. Since the marginals of the lottery over near-feasible ex-post allocations coincide with the random demand for each agent $i$, the marginals of the corresponding joint income distribution also coincide with $\randombudget{\agent}$ for all $\agent$. 

Two properties of the near-feasibility bound in Theorem~\ref{theo:feasibility} are worth emphasizing. First, the bound is tight. For example, when  $\Delta=1$ we recover the Birkhoff--von Neumann theorem.
Second, the bound is expressed as a good-by-good capacity violation that is independent of market size. This means that in order to ensure exact feasibility in any market, the designer simply has to reduce the capacities of each good by at most $\Delta-1$.\footnote{By contrast, the $\ell^2$-norm near-feasibility bound for ACEEI depends on the market size \citep{budish2011combinatorial}. Moreover, expressing possible capacity violations in the market-size-dependent $\ell^2$-norm makes it less practical for the designer to preemptively adjust capacity to ensure exact ex-post feasibility.}

\subsection*{Efficiency} We now analyze the efficiency properties of CERI both from the ex-ante and ex-post perspectives. We begin with the definition of ordinal efficiency.
\begin{definition}[\citeauthor{bogomolnaia2001new}, \citeyear{bogomolnaia2001new}]
Given an economy $\mathcal{E} = \economy$, a feasible lottery  allocation $\lottery{X}{}= (\lottery{x}{1},..,\lottery{x}{n})$  is \emph{ordinally efficient} if there does not exist another 
feasible  random  lottery allocation $\lottery{Y}{}= (\lottery{y}{1},..,\lottery{y}{n})$ such that 
$\lottery{y}{i} \succeq^{sd}_i \lottery{x}{i}$ for all $i$, and for at least for one agent the inequality is strict.
\end{definition}
A key result of our paper is that CERI characterizes ordinally efficient outcomes thereby  
\begin{theorem}[Characterization and Welfare Theorems]
\label{thm:characterization}
Given an economy $\mbox{$\mathcal{E} = \economy$}$, an allocation $\lottery{X}{}$ is ordinally efficient  if and only if there exists a profile of random incomes $(\randombudget{1},\ldots,\randombudget{\nagents})$ for which $\lottery{X}{}$ is a CERI allocation.  

\end{theorem}

The key insight in the proof lies in the requirement that for an allocation to be supported by a CERI, there must exist a price vector $\mathbf{p}$ such that, for every agent $i$, if the probability of consuming a bundle $\mathbf{x}$ is positive, then for every bundle $\mathbf{y} \succ_i \mathbf{x}$, the cost of $\mathbf{y}$ must be strictly greater than the cost of $\mathbf{x}$. Otherwise, agent $i$ would be better off consuming $\mathbf{y}$ instead of $\mathbf{x}$. Conversely, if such a price vector exists, it is possible to construct random budgets for the agents so that each lottery allocation precisely corresponds to the optimal consumption under the random budget.

Note that CERI prices can be scaled to ensure that the cost difference between bundles $\mathbf{y}$ and $\mathbf{x}$ is at least 1. Consequently, a lottery allocation can be sustained by a CERI if and only if there exists a price vector $\mathbf{p} \geq 0$ such that
 $$\prices \cdot (\y-\x)\ge 1 \text{ for every } \y\succ_i \x \text{ and } \prob(\tilde{\mathbf{x}}_i=\x)>0. $$ 

Therefore, the presence of a price vector $\mathbf{p}$ satisfying the given conditions can be expressed as a solution to a linear program. Farkas' lemma offers a characterization for the existence of such a solution through the dual of the linear program.  We show that the dual condition precisely captures the criterion for the allocation to be ordinally efficient.

 We now consider whether the allocations in the ex-post implementation of CERI are Pareto-efficient. As Theorem~\ref{theo:feasibility} showed ex-post implementations of CERI are only near-feasible in general. The following definition adapts Pareto efficiency to take into account the near-feasibility of allocations.

\begin{definition}\label{def:deltaeff}Given an economy $\mathcal{E} = \economy$, 
    a deterministic allocation $\allocs$ is \emph{$\kappa$-ex-post efficient} if there exists $\supplyvector'\le \supplyvector +\kappa\cdot {\bf 1}$, such that $\allocs$ is Pareto-efficient with respect to $\supplyvector'$.
\end{definition}

In words, $\kappa$-ex-post efficiency requires that when the supply of each good is relaxed by at most $\kappa$, the allocation is Pareto-efficient.
Each allocation in the CERI ex-post implementation corresponds to a competitive equilibrium under a particular realization of budgets and adjusted capacities. By the First Welfare Theorem, such a competitive equilibrium is Pareto-efficient with respect to the adjusted capacities. This logic is summarized in the following result.

\begin{theorem}[Ex-post efficiency]\label{theo:expostefficiency}
    All deterministic allocations in a $\Delta-1$-near feasible ex-post implementation of a CERI are $\Delta-1$-ex-post efficient.  
\end{theorem}

\begin{proof}
By Theorem~\ref{theo:feasibility}, every CERI has a $\Delta-1$-near-feasible ex-post implementation.
Using Definitions~\ref{def:implementation} and~\ref{def:nearimplementation}, note that each allocation in this ex-post implementation is a competitive equilibrium allocation with adjusted capacities. Therefore, by the First Welfare Theorem, each allocation is Pareto-efficient with respect to these adjusted capacities and hence $\Delta-1$-ex-post efficient.
\end{proof}

If the designer wished to ensure that the capacities are never violated ex-post, she can progressively reduce capacities, find a CERI and then check whether the allocations in the ex-post implementation violate the capacities. Theorem~\ref{theo:feasibility} and~\ref{theo:expostefficiency} guarantee can such preliminary capacity reductions never have to exceed $\Delta-1$ for any good.

\subsection*{Envy-freeness}

We now examine the envy-freeness properties of CERI both from the ex-ante and ex-post perspectives.

\begin{definition}[\citeauthor{bogomolnaia2001new}, \citeyear{bogomolnaia2001new}]
A lottery allocation $(\Tilde{\bf x}_1,.., \Tilde{\bf x}_n)$ is \emph{ordinally} \emph{envy-free} if
  $\Tilde{\bf x}_{i}\succeq^{sd}_{i} \Tilde{\bf x}_{i'}$ for all pairs of agents $i,  i'$.  
\end{definition}

In an ordinally envy-free lottery allocation no agent prefers the lottery of another agent. The Probabilistic Serial mechanism and the Bundled Probabilistic Serial mechanism are both ordinally envy-free, but the RSD and the ACEEI mechanism are not.
Since exact ex-post envy-freeness is not achievable in our setting due to the indivisibility of goods, we introduce the following relaxation based on envy-freeness up to one good due to \citet{lipton2004approximately} and \citet{budish2011combinatorial}.

\begin{definition} 
A lottery allocation $(\Tilde{\bf x}_1,.., \Tilde{\bf x}_n)$ is \emph{(ex-post) envy-free up to one good (EF1)}, if for all pairs of agents $i,  i'$  and every  realization ${\bf x}_i$ of $\Tilde{\bf x}_i$ and  ${\bf x}_{i'}$ of $\Tilde{\bf x}_{i'}$
there exists a good $j$ such that   ${\bf x}_i\succeq_i ({\bf x}_{i'}-{\bf e}^j)^+$. 
\end{definition}

Note that our definition of ex-post envy-freeness is a slight generalization of \citeauthor{lipton2004approximately}'s and \citeauthor{budish2011combinatorial}'s envy-freeness up to one good (which is satisfied by ACEEI allocations), because their definition applies to allocations rather than lottery allocations. Here, we extend envy-freeness up to one good to lottery allocations by requiring that every realization of the lottery allocation must satisfy EF1. 

The envy-freeness properties of CERI are summarized in the following result.

\begin{theorem}[Envy-freeness]\label{theo:envy}
Consider an economy $\mathcal{E}=\economy$ and a profile of random incomes $(\randombudget{1},\ldots,\randombudget{\nagents})$.
\begin{itemize}

\item[(i)] If for all agents $i, j \in N$, we have that $\randombudget{i} = \randombudget{j}$, %
then the CERI allocation is ordinally envy-free. \item[(ii)] If the support of all the budget distributions is within the interval $[b, \frac{\Delta}{\Delta-1}b]$, where $b > 0$ is a constant, then the ex-post implementation of the CERI is ex-post EF1. 
\end{itemize}
\end{theorem}

Note that the two conditions that ensure ex-ante and ex-post envy-freeness properties in Theorem~\ref{theo:envy} are independent. If income distributions are not the same, but have a sufficiently small support, it is possible that a CERI allocation is not ordinally envy-free, but each allocation in the ex-post implementation is EF1. Conversely, if income distributions are identical but come from a large interval, then the CERI allocation will be ordinally envy-free, but some (or all) allocations in the ex-post implementation would not be EF1.

\section{The CERI mechanism and relationships to existing mechanisms}
\label{sec:cerimech}
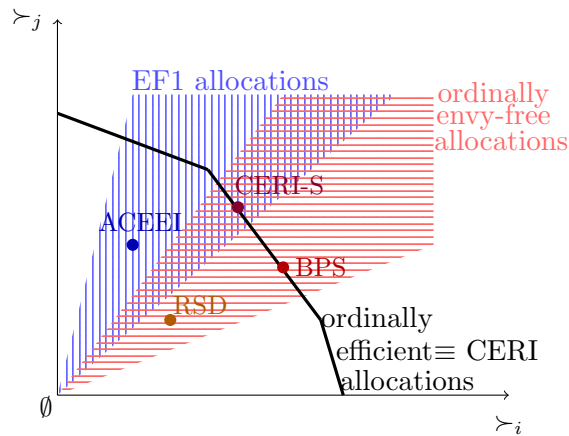
\begin{figure}
\smaller
\begin{centering}
\begin{tikzpicture}

  \fill[pattern=wide vertical lines, pattern color={rgb,255:red,102; green,102; blue,255}] (0,0) -- (1,4) -- (4.5,4)  -- cycle;
\fill[pattern=wide horizontal lines, pattern color={rgb,255:red,255; green,102; blue,102}](0,0) -- (3,4) -- (5,4) -- (5,2) -- cycle;

  \draw[->] (0,0) -- (6,0);
  \draw[->] (0,0) -- (0,5);

  \draw[very thick] (0,3.75) -- (2,3);
  \draw[very thick] (2,3) -- (3.5,1);
  \draw[very thick] (3.5,1) -- (3.8,0);

  \fill[orange!70!black] (1.5,1) circle[radius=0.08];
  \node[orange!70!black] at (1.9,1.2) {RSD};
  
  \node at (-0.15,-0.15) {$\emptyset$};
  \node at (-0.4,5) {$\succ_j$};
  \node at (6,-0.4) {$\succ_i$};
  \node at (4.25,1) {\small ordinally};
  \node at (4.35,0.6) {\small efficient};
  \node at (4.65,0.2) {\small allocations};
  \node at (5.7,0.6) {\small $\equiv$ CERI};
  
  \fill[blue!70!black] (1,2) circle[radius=0.08];
  \node[blue!70!black] at (1.1,2.3) {ACEEI};
  
  \fill[red!70!black] (3,1.7) circle[radius=0.08];
   \node[red!70!black] at (3.5,1.7) {BPS};
   
  \fill[purple!70!black] (2.4,2.5) circle[radius=0.08];
  \node[purple!70!black] at (2.95,2.8) {\CERIU};

  \node[blue!70] at (2.3,4.2) {\small EF1 allocations};

  \node[red!60] at (5.8,4) {\small ordinally};
  \node[red!60] at (5.8,3.7) {\small envy-free};
  \node[red!60] at (5.9,3.4) {\small allocations};
  
\end{tikzpicture}
\caption{Illustration of  the properties of the allocations of different mechanisms. Each point is a profile of lotteries for the two agents. More preferred lotteries are further from the origin. The frontier represents the set of ordinally efficient allocations, all of which can be achieved by a CERI. BPS = Bundled Probabilistic Serial. ACEEI = Approximate Competitive Equilibrium from Equal Incomes. RSD = Random Serial Dictatorship. %
ACEEI is a deterministic allocation near the frontier, with uncertain placement either inside or outside the frontier. Conversely, \CERIU~and BPS are on the frontier and can be precisely realized  with a lottery over allocations near the frontier.
\label{fig:tradeoffs}}
\end{centering}
\end{figure}

The properties described in the previous section allow us to define a range of mechanisms that implement a CERI. Abstractly, a \emph{CERI mechanism} maps an economy to a CERI allocation. Concretely, a CERI mechanism consists of the following steps: 
\begin{itemize}
\item[(i)] fix a profile of random incomes,
\item[(ii)] ask every agent $i\in N$ to report their $\succ_i$, 
\item[(iii)] compute a CERI, and 
\item[(iv)] construct a $\Delta-1$-near feasible ex-post implementation of this CERI.

\end{itemize}

A key parameter in any CERI mechanism is the profile of random incomes in step (i) because its choice governs the envy-freeness and incentive compatibility properties.

\subsection{\CERIU\ mechanism} For any economy, let \CERIU\ be a CERI mechanism in which every agent's income distribution is $U[1, 1+\epsilon]$ where $\epsilon< 1/m$. The following result summarizes the main properties of \CERIU.\footnote{Pronounced ``serious'' where ``S'' is for ``same and small income distributions''.}

\begin{theorem}[\CERIU]\label{thm:ceris}
The \CERIU\ mechanism is (i) ordinally efficient, (ii) $\Delta-1$-ex-post efficient, (iii) ordinally envy-free, (iv) ex-post envy-free up to one good, (v) strategyproof in the large.
\end{theorem}

\begin{proof}
Part (i) follows from  Theorem~\ref{thm:characterization}. Part (ii) follows from Theorem~\ref{theo:expostefficiency}. Parts (iii) and (iv) follow from Theorem~\ref{theo:envy} by setting $b=1$ and $\Delta\leq m+1$. Since \CERIU\ is ordinally envy-free, it also satisfies Definition~5 in \citet{azevedo2019strategy}. Hence, by Theorem~1 in \citet{azevedo2019strategy}, \CERIU\ is strategyproof in the large.
\end{proof}

The \CERIU\ mechanism is elementary because it simply aggregates the preferences, computes a CERI from them and outputs an ex-post implementation. \CERIU\ combines all the ex-ante and ex-post efficiency and envy-freeness properties discussed in the previous section. Moreover, \CERIU\ is strategyproof in the large \citep{budish2011combinatorial,azevedo2019strategy}, i.e., every agent has an incentive to report their preferences truthfully as a best response to CERI prices which she cannot affect in a large market.\footnote{We avoid a full technical definition of SP-L which can be found in \citet{azevedo2019strategy}.} 

\subsection{Relationship to existing mechanisms}\label{sec:relationship}
We discuss the relationship between CERI mechanisms and other mechanisms in the literature. Figure~\ref{fig:tradeoffs} provides an illustration.

\subsubsection*{Relationship to Serial Dictatorships}
A Serial Dictatorship is ex-post efficient (and ordinally efficient), but it is not ex-post envy-free up to one good.
The allocation of any given deterministic Serial Dictatorship can be implemented by a CERI mechanism in a straightforward manner. For instance, in the unit-demand case, set the  $k^\text{th}$ agent's budget to be $\$\frac{1}{k}$ w.p. 1. The CERI mechanism will set the price of the item that the $k^\text{th}$ agent consumes to $\frac{1}{k}$.

\subsubsection*{Relationship to the Random Serial Dictatorship} The Random Serial Dictatorship (RSD) selects one of Serial Dictatorship allocations uniformly at random. The RSD mechanism is ex-post efficient, but it is not ordinally efficient even in the unit-demand case, as the following example shows. 
\begin{example}[\citeauthor{bogomolnaia2001new}, \citeyear{bogomolnaia2001new}]\label{example:Bogo}
There are four unit-demand agents, $\{1,2,3,4\}$, and four goods, $\{a,b,c,d\}$.
Agents $1$ and $2$ have preferences: $\{a\} \succ_{1,2} \{b\} \succ_{1,2} \{c\} \succ_{1,2} \{d\}$, while agents $3$ and $4$ have preferences $\{b\} \succ_{3,4} \{a\} \succ_{3,4} \{d\} \succ_{3,4} \{c\}$. In the RSD, agents 1 and 2 each receive $a$ and $c$ w.p. $\frac{5}{12}$ and $b$ and $d$ w.p. $\frac{1}{12}$. Symmetrically, agents 3 and 4 each receive $b$ and $d$ w.p.  $\frac{5}{12}$ and $a$ and~$c$ w.p. $\frac{1}{12}$. This lottery allocation not ordinally efficient because agent 1 prefers to trade their probability of receiving good~$b$ with agent 3 for a higher probability of receiving~$a$.
\end{example}
Theorem~\ref{thm:characterization} implies that the RSD lottery allocation in Example~\ref{example:Bogo} cannot be supported by a CERI. Intuitively, since the price of each item can vary substantially across different Serial Dictatorships, achieving coordination on the same prices across every Serial Dictatorship, as required by a CERI, is impossible. Finally, the RSD is neither ordinally envy-free \citep[p. 307]{bogomolnaia2001new} nor ex-post envy-free up to one good \citep[p. 1064--1065]{budish2011combinatorial}. %

\subsubsection*{Relationship to ACEEI}
The ACEEI mechanism is approximately ex-post efficient and ex-post envy-free up to one good.
However, similarly to the RSD, the ACEEI mechanism might be neither ordinally efficient nor ordinally envy-free. In the unit-demand case the ACEEI mechanism coincides with the RSD. To see this, consider Example~\ref{example:Bogo} again.  In the ACEEI mechanism \cite[pp. 1080--1081]{budish2011combinatorial}, agents' budgets are uniformly perturbed within the interval $[1-\epsilon, 1]$, and a competitive equilibrium is then computed based on these adjusted budgets. In the case of unit demand, consider ordering the agents by their budgets from largest to smallest and allocating items according to a Serial Dictatorship. This allocation corresponds to a CERI where the price of each item matches the budget of the agent who receives it (see above).  Because the perturbations are uniformly random, the ACEEI mechanism produces a lottery allocation identical to the one produced by the RSD. Hence, the ACEEI mechanism is not ordinally efficient. In Appendix~\ref{app:accei}, we give another example of the ordinal inefficiency of the ACEEI mechanism in the multiunit-demand setting and in Appendix~\ref{app:B} we show that the existence of a CERI implies the existence of an ACEEI.

\subsubsection*{Relationship to Simultaneous Eating and the Probabilistic Serial mechanisms}
Simultaneous Eating mechanisms and in particular the Probabilistic Serial (PS) mechanism, were introduced and algorithmically defined by \citet{bogomolnaia2001new} for the unit-demand setting. A Simultaneous Eating mechanism involves an ``eating'' procedure in which agents ``eat'' fractional amounts of the most preferred available item at a certain ``speed''. Once an entire item has been ``eaten'', agents proceed to ``eat'' the next most preferred available item. Since Simultaneous Eating mechanisms characterize ordinally efficient allocations for the unit-demand setting,
each lottery allocation produced by a Simultaneous Eating mechanism can be supported by a CERI.

There is a formal and intuitive mapping between Simultaneous Eating mechanisms and CERI mechanisms: In any Simultaneous Eating mechanism, at any time $t$, an agent always ``eats'' the most preferred available item while in the corresponding CERI mechanism, the agent receives the most preferred item at a price at most $t$.  Concretely, consider the following mapping from a Simultaneous Eating mechanism to a CERI mechanism. First, normalize the ``eating speeds'' so that the Simultaneous Eating mechanism runs as a descending clock, starting at time 1 and finishing at time 0. Without loss of generality, we can also assume that each agent ``eats'' exactly 1 unit of the goods (e.g., by adding a dummy good).  The first moment at which an item is fully ``eaten'' corresponds to the item’s price in the CERI mechanism. If an item is not fully ``eaten'' by the end of the mechanism, its price will be 0. The budget distribution of agents is defined on the interval $[0,1]$, and the ``eating speed'' of an agent at time $t$ corresponds precisely to the density of the budget distribution at $t$.

The PS mechanism is a special case of Simultaneous Eating mechanisms where the ``eating speed'' is the same for all agents. In addition  to being ordinally efficient, the PS mechanism is ordinally envy-free.
Hence, the PS mechanism corresponds to a CERI mechanism with uniform budget distributions $U[0,1]$ for all agents and prices constructed as just described above. Hence, our Theorem~\ref{theo:envy} immediately implies that the PS mechanism must be ordinally envy-free.

\subsubsection*{Relationship to the Bundled Probabilistic Serial mechanism}
\citet{nguyen2016assignment} introduced and algorithmically defined a generalization of the PS mechanism for the multiunit-demand setting called the Bundled  Probabilistic Serial (BPS) mechanism. Rather than ``eating'' their most preferred goods as in PS mechanism, in the BPS mechanism agents ``eat'' the most preferred available \emph{bundle} of goods. \citet{nguyen2016assignment} demonstrated that a BPS lottery allocation is ordinally efficient. We first illustrate how to map the BPS mechanism to a CERI mechanism and then show that BPS can only implement  a strict subset of ordinally efficient outcomes.

The mapping from the BPS mechanism to a CERI mechanism is similar to the mapping of the Simultaneous Eating mechanism to the CERI mechanism in the unit-demand case, but with one modification in the calculation of prices. This modification ensures that, for each agent, the price of a bundle ``eaten'' earlier is strictly higher than that of a bundle ``eaten'' later. This property of bundle prices ensures that one can construct a budget distribution so that the consumption of an agent in the ``eating'' procedure coincides with the consumption under CERI.

To achieve the required property of bundle prices, we arrange the items based on the order in which they are fully ``eaten''. For item $i$, let $r_i$ represent its order, where $r_i=1$ if it is the first item and $r_i=\infty$ if the item is still available at the end of the mechanism. Let $K$ be the largest size of an acceptable bundle, and set the price of an item to $\frac{1}{(K)^{r_i}}$. If a bundle is available before item $i$ is fully ``eaten'', its price is at least $\frac{1}{(K)^{r_i}}$. On the other hand, if a bundle $x$ remains available after item $i$ is fully ``eaten'', then the price of $x$ is less than $K \times \frac{1}{(K)^{r_i+1}}= \frac{1}{(K)^{r_i}}$.

Unlike CERI, Bundled Simultaneous Eating mechanisms (i.e., the BPS mechanism with heterogeneous ``eating speeds'') do not characterize the set of all ordinally efficient allocations in the multiunit-demand setting. The following example illustrates that Bundled Simultaneous Eating mechanisms produce allocations that are a strict subset of CERI allocations.
\begin{example} \label{ex:eating}
Consider two agents $\{1,2\}$ and two goods $\{a,b\}$. The preferences of Agent 1 are  $\{a,b\}\succ_1 \{a\}\succ_1 \{b\} \succ_1 \emptyset$, while Agent 2's preferences are $\{a,b\}\succ_2 \{b\}\succ_2 \{a\} \succ_2 \emptyset$.

In any Bundled Simultaneous Eating mechanism, an agent always ``eats'' the best bundle available according to their preferences. Since both agents have the same most preferred bundle, any Bundled Simultaneous Eating mechanism will result in both agents' receiving the bundle $\{a,b\}$ with a positive probability. In the BPS, each agent receives $\{a,b\}$ w.p. $\frac{1}{2}$.

However, the allocation that assigns good ${a}$ w.p. 1 to Agent 1 and good ${b}$ w.p. 1 to Agent 2 is also ordinally efficient, but it cannot be obtained in a Bundled Simultaneous Eating mechanism under any profile of ``eating speeds''.

Both of the above allocations are ordinally efficient and hence both can be supported by a CERI. For example, the allocation (resulting from the BPS mechanism) where each agent receives bundle $\{a, b\}$ w.p. $\frac{1}{2}$ can be supported by CERI prices $p_a = p_b = 1$ and {budget distributions $\mathcal{B}_1 = \mathcal{B}_2 = (\$2\, \text{ w.p. }\, \frac{1}{2}; \$0\, \text{ w.p. }\,\frac{1}{2})$}. On the other hand, the second allocation can be supported by CERI prices $p_a = p_b = 1$ and (degenerate) budget distributions {$\mathcal{B}_1 = \mathcal{B}_2 = (\$1\, \text{ w.p. }\, 1)$}.

\ignore{
a bit more on this example: give agents budget uniform randomly between $[c_1,c_2]$
\begin{itemize}
    \item if $c_1\ge c_2/2$, then equilibrium is $c_2/2\le p\le c_1$, allocation is a to 1 b to 2.
    \item if $c_1 < c_2/2$,   $p=\frac{c_1+c_2}{3}$. Allocate $\{a,b\}$ to each with prob $\frac{c_2-2c_1}{3(c_2-c_1)}$, a to 1, b to 2 with prob $\frac{c_1+c_2}{3(c_2-c_1)}$
    
\end{itemize}
}
    
\end{example}

\section{Incentive properties} %
\label{sec:strategyproof}
It remains to discuss the incentive properties of CERI mechanisms. To formally describe incentives in a large market, we fix the set $M$ of goods and the size $\Delta$ of the largest bundle,  increase the number of agents, but allow capacities and agents' preferences to be arbitrary.
The \CERIU\ mechanism has already offered a strategyproof-in-the-large (SP-L) implementation of CERI (Theorem~\ref{thm:ceris}). However, SP-L only evaluates deviation incentives from an interim perspective: it merely requires that truthtelling be a best response to the empirical distribution of opponent types which must become exogenous to any given agent in a large market. 
Another large-market incentive guarantee introduced by \cite{liu2016ordinal} requires that each agent find it approximately optimal to report their type truthfully in response to \emph{any realization} of other agents' reports.
Denote by $\mathbb{P}_\mathbf{y}\left(\Phi_i(\capacities,\succ)\right)$ the probability that agent $i$ obtains (an acceptable) bundle $\mathbf{y}$ under the mechanism $\Phi$, when the reported preference profile is $\succ$ and the capacities are $\capacities$.
In our setting with ordinal preferences, such a notion of \emph{asymptotic strategyproofness} is defined as follows.\footnote{Other existing notions of asymptotic strategyproofness are weaker as they are expressed in terms of von Neumann–Morgenstern utilities~\citep{kojima2010incentives,hatfield2018strategy,azevedo2019strategy}.}

\begin{definition}\label{def:liu}
A mechanism $\Phi$ is \emph{asymptotically strategyproof} if for every $\eta>0$, there exists $\nthresh$ such that if the number of agents is $n >\nthresh$, then for all agents $i$, preference profiles $\{\succ_i\}_{i=1}^n$, capacity vectors $\capacities$ and reports $\succ'_i$:
$$
\sum_{\mathbf{y}\succeq_i \mathbf{x}}\mathbb{P}_\mathbf{y}\left(\Phi_i(\capacities,(\succ_i,\succ_{-i}))\right)\geq \sum_{\mathbf{y}\succeq_i \mathbf{x}}\mathbb{P}_\mathbf{y}\left(\Phi_i(\capacities,(\succ'_i,\succ_{-i}))\right)-\eta \;\; \text{for all}\;\; \mathbf{x}\in\Psi_i.
$$
\end{definition}

In words, asymptotic strategyproofness requires that every agent's lottery under truthtelling ``almost'' first-order stochastically dominates the lottery under any misreport in a large market.

Asymptotic strategyproofness is strictly stronger than SP-L and indeed many mechanisms used in practice that are SP-L (such as uniform-price auctions and stable matching mechanisms) are not asymptotically strategyproof \citep[Section 3.3.2]{azevedo2019strategy}. 
However, all notions of incentive compatibility in large markets (including SP-L and asymptotic strategyproofness) share a fundamental limitation: they bound misreporting incentives \emph{agent by agent} and merely assume that each agent is content with following an approximately optimal strategy. Such guarantees are not explicitly tied to how the mechanism is implemented and do not ensure that, with some probability, the mechanism is strategyproof. In fact, there is no guarantee that any positive fraction of the agents will find it optimal to report their preferences truthfully with any given probability. This is a significant shortcoming of existing incentive guarantees, since a mechanism’s outcomes may deviate substantially from theoretical predictions even if only a small fraction of agents misreport their preferences.
To overcome this limitation of existing incentive guarantees, we will require a stronger property: we say that a mechanism is \emph{uniformly strategyproof} if the event that \emph{all} agents have a dominant strategy to report their types truthfully occurs with a high probability in a large market.

In order to describe a uniformly strategyproof mechanism based on CERI, we will use an additional external source of randomness in the mechanism. In particular, let $\samplespace$ be a (finite) sample space, and for every outcome $\outcome\sim \samplespace$, let   $\Phi^\outcome$ be a direct mechanism. Define $\Phi^{\samplespace}$ to be the \emph{random mechanism} that  outputs $\Phi^\outcome$ with the same probability as drawing $\outcome$ from $\samplespace$. Therefore, $\mathbb{P}_{\outcome\sim \samplespace}\left(\Phi_i^\omega(\capacities,\succ)\right)$ denotes the lottery that agent $i$ obtains under the random mechanism $\Phi^\Omega$, when the reported preference profile is $\succ$, the capacities are $\capacities$ and the realization of the external source of randomness is $\omega$.
The following definition formalizes our new guarantee for truthtelling incentives.

\begin{definition}~\label{def:uniform}
A random mechanism $\Phi^{\samplespace}$ is \emph{uniformly strategyproof} if for every $\epsilon>0$, there exists $\nthresh$ such that if the number of agents is $n >\nthresh$, then for all  preference profiles $\{\succ_i\}_{i=1}^n$ and capacity vectors $\capacities$, and reports $\succ'_i$:
$$
\mathbb{P}_{\outcome\sim \samplespace}(\Phi_i^\outcome(\capacities,(\succ_i,\succ_{-i}))\succeq^{sd}_i \Phi_i^\outcome(\capacities,(\succ'_i,\succ_{-i})) \text{ for all } i \text{ and } \succ_i' )\ge 1-\epsilon.
$$

\end{definition}

Definition~\ref{def:uniform} of uniform strategyproofness is stronger than Definition~\ref{def:liu} of asymptotic strategyproofness in two ways. First, though less importantly, if a random mechanism is uniformly strategyproof, then for any agent $i$ and any preference profile $\succ_{-i}$ of the other agents, the probability that the lottery achieved by a misreport first-order stochastically dominates her lottery under truthtelling is at most $\epsilon$. This immediately implies that the error bound $\eta$ in Definition~\ref{def:liu} of asymptotic strategyproofness is bounded by $\epsilon$.
Second, and more importantly, uniform strategyproofness ensures that with arbitrarily high probability the mechanism is strategyproof; that is, it is a dominant strategy for \emph{every agent} to be truthful in a large market. This immediately implies that any individual agent is better off by being truthful in a large market which is all that is required by Definition~\ref{def:liu}.
Thus, uniform strategyproofness strictly strengthens both asymptotic strategyproofness and, \emph{a fortiori}, \mbox{SP-L}.

\subsection{Random grid}
We will now show how we can use CERI to implement a large-market mechanism that is uniformly strategyproof but nevertheless maintains the strong guarantees for efficiency and envy-freeness. 
The key feature that allows the mechanism to achieve its desirable properties is an external source of randomness that we call a \emph{random grid}.
In our next CERI implementation, we will approximate the reported type distribution by a nearby point on the grid. The intuitive reason for doing this is that with a sufficiently coarse grid, the likelihood of a single agent significantly altering the approximation to the grid is reduced, giving agents stronger incentives to report their types truthfully in a large market. On the other hand, the grid cannot be too coarse because in that case the grid point we use to approximate the type distribution might be too far from the true distribution, leading to an inefficient allocation.  
Let $\tau\in \mathbb{Z}_{>0}$ denote the dimension and $\gridjump\in \mathbb{Z}_{>0}$ denote the step size.
\begin{definition}\label{grid}
The \emph{random grid} $\mathbf{G}(\gridjump, \tau)$ is defined as follows. First, for each coordinate $t \in \{1, \ldots, \tau\}$, draw a random point $z^t$ independently and uniformly from $\{1, \ldots, \gridjump\}$. Second, set
\(
\mathbf{G}^t = \{0, \, z^t, \, z^t + \gridjump, \, z^t + 2\gridjump, \ldots\}.
\)
The random grid is then:
\[
\mathbf{G}(\gridjump, \tau) \;=\; \mathbf{G}^1 \times \mathbf{G}^2 \times \cdots \times \mathbf{G}^\tau.
\]
\end{definition}

The random grid is instantiated in the $\tau$-dimensional type-space, where dimension $\tau$ is the number of agent types. Hence, a point on the grid represents the number of agents of each type. The step size $\lambda$ of the grid is key to trading off the properties of the mechanism: larger $\lambda$ can improve truthtelling incentives but can reduce efficiency. We note that if the grid were deterministic, certain type distributions would tend to be near the grid lines. In such cases, the choice of grid point for the approximation would become highly sensitive to the actions of a single agent. Choosing a random grid mitigates this problem and ensures that, for every type distribution, the probability of any single agent substantially altering the approximation in a large market is negligible.

\subsection{\Mone\ mechanism}
Our large-market mechanism, denoted \Mone, works as follows.\footnote{Pronounced ``cereal'' where ``L'' is for ``large market''.} First, it instantiates a random grid. It then elicits a vector of agent types $\nagtypevec = (\nagtype^1, \nagtype^2, \ldots, \nagtype^{\numtypes })$, where  $\nagtype^t$ denotes the number of agents of type $t$, and approximates $\nagtypevec$  by a grid point $\typeapprox$. The mechanism then computes a CERI on the economy defined by $\typeapprox$, by identical budget distributions with a small support, and by the true capacities, and then implements an ex-post allocation according to this CERI. For clarity, we use an alternative definition of an economy. Instead of $\mathcal{E} = \economy$, we define it using the capacities and a vector of agent types: $\mathcal{E} = \pseconomy$, where $\nagtypevec \in \mathbb{Z}^\numtypes $ is as above. 
Formally, the mechanism is the following.

\begin{definition}[\Mone~mechanism] Let $\mathcal{E} = \pseconomy$ denote the economy, let $\typeset$ denote the set of agent types, let  $\numtypes  = |\typeset|$, and let $\grid{\gridjump,\numtypes }$ be a random grid for some $\gridjump\in \mathbb{Z}_{>0}$.
\begin{enumerate}
\item[(i)] Fix each agent's budget distribution to be $\randombudget{\agent}\sim U[1,\frac{\Delta}{\Delta-1}]$.
\item[(ii)] Elicit the agents' types. %
Let $\nagtypevec = (\nagtype^1, \nagtype^2, \ldots, \nagtype^{\numtypes })$, where  $\nagtype^t$ denotes the number of agents of type $t \in \typeset$. Let $\typeapprox$ be the minimal grid point that is coordinate-wise at least  $\nagtypevec$. Formally, \\      
    $\overline{\psi}^t=\min \{g \in G^t: g \geq \psi^t\}$ for each $t=1, \ldots, \tau$ and let $\typeapprox=\left(\overline{\psi}^1, \ldots, \overline{\psi}^\tau\right)$.
    
    \item[(iii)] Compute a CERI for   $\mathcal{E}' = \zeconomy$. %
    \item[(iv)] 
    Construct a $\Delta-1$-near feasible ex-post  implementation of this CERI (according to Theorem~\ref{theo:feasibility}).
\end{enumerate}

\end{definition}
\Mone\ specifies two features of a CERI mechanism. First, like \CERIU, it sets agents' budget distributions to be identical and to have a small support. Second, \Mone\ approximates the type distribution using a random grid before calculating a CERI.
Figure~\ref{fig:grid} illustrates how the random grid works in \Mone. %
In this illustration, there are two types: $\tau=2$. We elicit agents' types $\nagtypevec = (\nagtype^1, \nagtype^2)$ and  then compute the minimal grid point $\typeapprox$ that, for each agent type $t$, is no smaller than the stated number of types, $\nagtype^t$. We then use these approximate type vector as an input in \Mone. %

\begin{figure}[t]
\smaller
\begin{centering}
\begin{tikzpicture}[scale=1.3]

  \foreach \x in {0.6,1.6,...,5.6} {
    \draw[thin, dotted] (\x,0) -- (\x,5.3);
  }
  \foreach \y in {0,1,...,5} {
    \draw[thin, dotted] (0,\y+0.3) -- (5.6,\y+0.3);
  }
  \draw[-] (0,0) -- (5.6, 0);
  \draw[-] (0,0) -- (0,5.3);

  \node at (-0.1,-0.1) {0};
  \node[rotate=90, above] at (-0.8,3) {number of agents of type 2};
  \node[below] at (3,-0.4) {number of agents of type 1};

  \draw[fill=gray] (2.8,3.9) circle[radius=0.06];
  \node at (3.05,3.75) {$\nagtypevec$};

  \iftoggle{version}{
  \draw[fill=gray] (2.5,3.9) circle[radius=0.06];
  \node at (2.2,3.75) {$\nagtypevec^{*}$};
  \draw[fill=black] (2.6,4.3) circle[radius=0.06];
  \node at (2.2,4.45) {$\typeapprox^{*}$};
  }{}
  
  \draw[fill=black] (3.6,4.3) circle[radius=0.06];
  \node at (3.75,4.45) {$\typeapprox$};

 \node at (0.6,-0.2) {$z^1 $};
 \node at (-0.2,0.3) {$z^2 $};
  \node at (1.6,-0.2) {$z^1 + $ $\gridjump$};
  \node at (-0.4,1.3) {$z^2 + $ $\gridjump$};

  \foreach \i in {2,3,...,5} {
    \node at (\i + 0.6,-0.2) {$z^1 + $ \i $\gridjump$};
    \node at (-0.4,\i+0.3) {$z^2 + $ \i $\gridjump$};
    \foreach \x in {0,1,2,3,4,5} {
        \foreach \y in {0,1,2,3,4,5} {
        \filldraw (\x+0.6,\y+0.3) circle [radius=0.01];
        }
    }
  }
\end{tikzpicture}
\caption{Example of a realization of a random grid $\mathbf{G}(\gridjump, \numtypes )$, where $\numtypes =2$, and the grid is seeded with the point $(z^1, z^2)$. Also depicted is the rounding used in Mechanism \Mone, where two  type vectors ($\psi$ and $\psi^*$) are approximated by two grid points ($\overline{\psi}$ and $\overline{\psi}^*$) respectively.} 
\label{fig:grid}
\end{centering}
\end{figure}
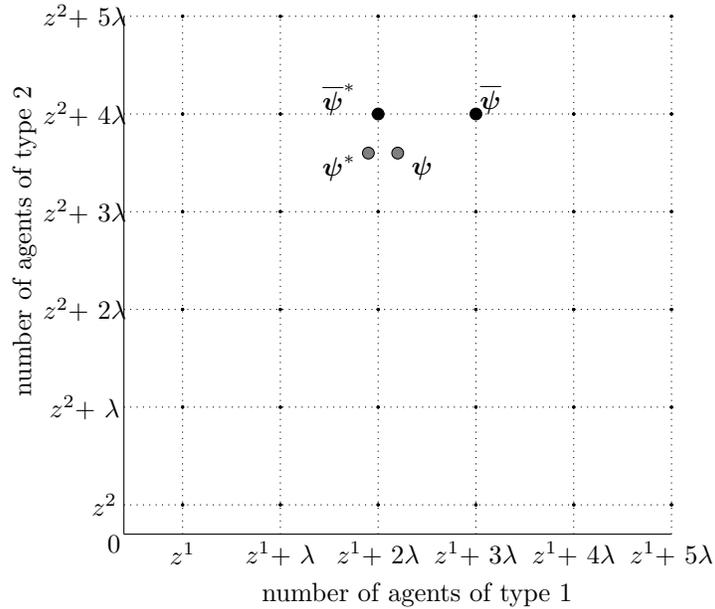

\subsection{Properties}\label{sec:asymproperties}

Since \Mone~ approximates agents' preferences by the random grid, it will not in general be ordinally efficient. Nevertheless, we can guarantee efficiency properties of \Mone~ in large markets.
In order to define notions of asymptotic efficiency, we will require the following two definitions of ex-ante and ex-post approximate efficiency.

\begin{definition}
Given an economy $\mathcal{E} = \economy$,     a lottery allocation $\Tilde{\mathbf{X}}$ is  \emph{ordinally $\epsilon$-efficient} for $\mathcal{E}$ if there exists some $\capacities'$, $\capacities(1-\epsilon)\le \capacities'\le \capacities(1+\epsilon)$, such that $\Tilde{\mathbf{X}}$ is ordinally efficient for the economy $\mathcal{E}' = \alteconomy$.   
\end{definition}

Ordinal $\epsilon$-efficiency simply requires that an allocation be ordinally efficient in an economy with slightly relaxed capacities.
Naturally, we also consider a definition of approximate ex-post efficiency. 
\begin{definition}
Given an economy $\mathcal{E} = \economy$,     a deterministic allocation ${\mathbf{X}}$ is  \emph{$(\kappa, \epsilon)$-ex-post  efficient} if there exists some $\capacities'$, $\capacities(1-\epsilon)\le \capacities'\le \capacities(1+\epsilon)$, such that ${\mathbf{X}}$ is $\kappa$-ex-post  efficient for the economy $\mathcal{E}' = \alteconomy$.  \end{definition}

A $(\kappa, \epsilon)$-ex-post efficient allocation relaxes capacities in two ways: by a multiplicative factor of $1\pm \epsilon$ and by an additive error of at most $\kappa$. As above, we will show that the additive error $\kappa$ is bounded by the maximum bundle size minus one. Hence, in large markets, the additive error is negligible relative to capacities and can be absorbed into the multiplicative error.

We can now define our two notions of \emph{asymptotic efficiency}, one from the ex-ante perspective and the other from the ex-post perspective. Let $\typedist$ be the (discrete) distribution of agents' types, and let $\typedist^n$ denote the distribution of preference profiles for $n$ agents drawn i.i.d.\ from $\typedist$.

\begin{definition}\label{def:strongAE}
 A random mechanism $\Phi^{\samplespace}$ is  \emph{asymptotically {ordinally} efficient} if for every $\epsilon$, and distribution $\typedist$ of agent types, there exists $\nthresh$ such that for all $\capacities$ and $n > \nthresh$  
 $$
 \prob_{\succ \sim  \typedist^n} (\Phi^\outcome\economy \text{ is ordinally  $\epsilon$-efficient})\ge 1-\epsilon  \text{ for all } \outcome\sim \samplespace.
 $$
 \end{definition}

 \begin{definition}
A random mechanism $\Phi^{\samplespace}$ is \emph{asymptotically $\kappa$-ex-post efficient} if for every $\epsilon$, and distribution $\typedist$ of agent types, there exists $\nthresh$ such that for all $\capacities$ and $n > \nthresh$ 
$$
 \prob_{\succ \sim  \typedist^n}(\text{all allocations produced by } \Phi^\outcome \text{ are } (\kappa, \epsilon)\text{-ex-post efficient}) \ge 1-\epsilon \;\;  \text{ for all } \outcome\sim \samplespace.
 $$
         
\end{definition}

In an asymptotically  ordinally/$\kappa$-ex-post efficient mechanism, the ex-ante/ex-post allocation is guaranteed to be ordinally $\epsilon$-/($\kappa,\epsilon)$-ex-post efficient with an arbitrarily high probability when the market is sufficiently large. Note that the probability bound has to hold in \emph{all} outcomes of the random grid.

As the agents' budget distributions are identical and meet the requirements of Theorem~\ref{theo:envy}, \Mone~is ordinally envy-free and ex-post EF1 despite the approximation of types by the random grid. In addition, \Mone~is uniformly  strategyproof  and asymptotically efficient from ex-ante and ex-post perspectives as the following result shows.

\begin{theorem}[\Mone]\label{thm:m1}
The \Mone~mechanism with random grid step size $\mbox{$\gridjump = \left\lfloor\numtypes \sqrt{n} \right\rfloor$}$ is (i) uniformly  strategyproof, (ii) asymptotically ordinally efficient,  (iii) asymptotically $\Delta-1$-ex-post efficient,  (iv) ordinally envy-free and (v) ex-post EF1.
\end{theorem}
The proof is in Appendix~\ref{proof:thm1}.  
The choice of the random grid step size $\gridjump$ is crucial as $\gridjump$ represents a trade-off between efficiency and strategyproofness. When the grid step size $\gridjump$ is set to 1, \Mone\ coincides with \CERIU\. In this case, \Mone\ is ordinally efficient in a market of any size, but it cannot be guaranteed to be uniformly strategyproof.
A larger $\gridjump$ reduces the probability of an agent having a significant influence on the grid point and, consequently, on equilibrium prices, {thereby increasing truthtelling incentives}. 
However, a larger $\gridjump$ also corresponds to a larger error in the approximation of the true type distribution, leading to a less efficient allocation. As Theorem~\ref{thm:m1} shows, in a large market, one can {carefully} choose a grid step size to achieve a high probability of both minimizing the influence of individual agents on the grid and ensuring an asymptotically ordinally efficient allocation.

\section{Conclusion}\label{sec:conclusion}
This paper has laid price-theoretic foundations for efficient combinatorial assignment.
Our characterization of ordinally efficient allocations as CERI allocations provides valuable insights for comparing existing mechanisms and for constructing new ones. Our results illuminate the relationships between efficiency, envy-freeness and incentive compatibility in combinatorial assignment settings and suggest ways to manage the tradeoffs between them in market design. 

There are many possible research directions. 
The first direction is to theoretically explore other implementations of CERI that might strike a different balance between the objectives of market designers or to add new features to CERI. For example, one can imagine incorporating priorities into CERI \citep{kornbluth2021undergraduate}.
The second direction would be to give formal results about the complexity of computing (an approximate) CERI and to develop algorithms for practical implementation (cf. \citet{othman2016complexity}, \citet{budish2017course} and \citet{budish2023practical} for such results for the ACEEI).
The third direction would be to examine how CERI performs empirically relative to other mechanisms. For example, \citet{budish2012multi} empirically examined the Harvard Business School ``draft'' mechanism and pointed out its properties relative to RSD. %
\citet{bichler2021randomized} also examined other combinatorial assignment mechanisms in the course allocation setting, and showed that BPS in particular outperforms first-come-first-served mechanisms. Since CERI improves on the properties of both ACEEI and BPS, these initial results suggest that CERI would perform well in practice.

\newpage
\appendix
\renewcommand{\thelemma}{\thesection.\arabic{lemma}}
\begin{center}
{\bf \large APPENDIX}\\   
\end{center}
\vspace{0.4cm}

\section{Omitted proofs}
\subsection{Proof of Theorem~\ref{theo:exist}}\label{app:exist}
We use a standard fixed-point argument.
First, note that since the cumulative distribution functions of budgets are continuous, the expected consumption $\expect{\optbundle{\agent}(\prices,\randombudget{\agent})}$ is continuous with respect to $\prices$ for each agent $j$.

Second, as the price of a good approaches infinity, the probability of demanding a bundle containing that good tends to $0$. Therefore, there exists a positive constant $P$ such that if the price of good $j$, denoted as $\price{j}$, exceeds $P$, then the expected demand for good $i$ from every agent is less than $\supply{j}/n$. %
This implies that the aggregate consumption of good $j$ will be strictly less than the supply when the price exceeds $P$.

Recall that  $\choices_\agent$ is the set of acceptable bundles available to agent $i$, and   $\mathcal{L}(\choices_\agent)$ is  the set of lotteries over $\choices_\agent$.
In the proof, we will construct a mapping 
 $$f:[0,P]^m\times \mathcal{L}(\choices_1)\times \ldots \times \mathcal{L}(\choices_n)\rightarrow [0,P]^m\times \mathcal{L}(\choices_1)\times \ldots \times \mathcal{L}(\choices_n)$$
and use Brouwer's Fixed-Point Theorem to conclude that it has  a fixed point. We will then show that this fixed point corresponds to a CERI.

To begin, given a lottery allocation $\{\lottery{x}{i}\in \mathcal{L}(\choices_i) \}_{i=1}^n$,   let $\overline{\bf{x}}=\expect{\sum_{i=1}^n \lottery{x}{i}}$ be the aggregate expected consumption. The excess demand vector is $\overline{\bf{x}}-\supplyvector$.
Let 

$$
z_j(\prices, \lottery{x}{1},..,\lottery{x}{n}):= \min\{(\price{j}+ \overline{x}^j-\supply{j})^+,P\} \text{ for all } j\in M,
$$
and denote $\z:= (z_1,..,z_m)$.
Define the following  mapping
$$
f(\prices,\lottery{x}{1},..,\lottery{x}{n}):=(\z, \optbundle{1}(\prices,\randombudget{1}), ...,\optbundle{n}(\prices,\randombudget{n})).
$$
It is easy to see that $f$ satisfies all the conditions (continuity and boundedness) of Brouwer's Fixed-Point Theorem. Let $(\prices,\lottery{x}{1},..,\lottery{x}{n})$ be a fixed point of the mapping.  

First, at this fixed point $\lottery{x}{i}= \optbundle{i}(\prices,\randombudget{i})$ for all $i \in [n]$. It remains to check the market-clearing condition (ii) in Definition~\ref{def:ceri} of CERI.

Since the mapping is defined on the set of prices between $0$ and $P$, we have $\price{j}\le P$. We will show that the strict inequality holds, that is,  $ \price{j}<P$ for all goods $j \in M$. If this is the case then the fixed-point condition implies that $ \price{j}=z_j=(\price{j}+\overline{x}^j-\supply{j})^+$. And therefore if $\price{j}>0$, then $\overline{{x}}^j= \supply{j}$ and if $\price{j}=0$, then $\overline{{x}}^j \le \supply{j}$ which is the market-clearing condition (ii) in Definition~\ref{def:ceri}.

Finally, assume towards a contradiction that $\price{j} = P$ for a good $j$, then by the way $P$ is selected, 
the aggregate expected demand for good $j$, $\overline{x}^j$, will be strictly less than $\supply{j}$. Therefore, 
$$
z_j(\prices, \lottery{x}{1},..,\lottery{x}{n})=\min\{(\price{j}+ \overline{x}^j-\supply{j})^+,P\}= \min\{(P+ \overline{x}^j-\supply{j})^+,P\}<P = \price{j},
$$
which contradicts the fixed-point condition and completes the proof.
\qed

\subsection{Proof of Theorem~\ref{theo:feasibility}} 
Let $(\prices, \lottery{X}{})$ be a CERI. To construct a near-feasible ex-post implementation of CERI, it suffices to express the lottery allocation $\lottery{X}{}$ as the marginal distributions of a lottery over near-feasible ex-post allocations.   

Consider an agent $i$. Let $x_i$ be a realization of the random demand $\optbundle{\agent}(\prices, \randombudget{\agent})$. Define $\randombudget{\agent}|_{x_i}$ as the conditional income distribution such that the optimal consumption is $x_i$. This distribution has the density function  

\[
f_{\randombudget{\agent}|_{x_i}}(b) := \frac{f_{\randombudget{\agent}}(b)}{\prob(\optbundle{\agent}(\prices, \randombudget{\agent}) = x_i)}
\]
if agent $i$ consumes $x_i$ under income $b$, and $0$ otherwise.  

We use the following result from \citet[Theorem 2.1]{nguyen2016assignment}:  

\begin{quote}
Any feasible random allocation with respect to $\supplyvector$ can be realized through deterministic allocations that are feasible with respect to $\supplyvector + (\Delta - 1) \cdot {\bf 1}$.
\end{quote}

As a consequence, there exists a lottery $\Tilde{A}$  over $\Delta - 1$-near-feasible deterministic allocations such that its $i$-th marginal distribution is $\lottery{X}{i}$. The joint distribution of income is then constructed as follows: for each $(x_1, \dots, x_n)$ drawn from $\Tilde{A}$, we generate $(\randombudget{\agent}|_{x_1}, \dots, \randombudget{\agent}|_{x_n})$, where the coordinates are independent distributions.  
This construction of the joint income distribution, together with $\Tilde{A}$, provides a $\Delta - 1$-near-feasible ex-post implementation of the CERI. \qed

\subsection{Proof of Theorem~\ref{thm:characterization}.}\label{app:maintheorem}

 If a feasible allocation $\lottery{X}{}=(\lottery{x}{1},\ldots,\lottery{x}{n})$ is CERI, then there exists a price vector $\prices$ that satisfies the following two conditions. 
 
The first condition is the market-clearing condition:
\begin{equation}\label{eq:c1}
    p_j \ge 0; \quad p_j = 0 \text{ if  good $j$ is not fully allocated, i.e., } (\sum_i \expect{\lottery{x}{i}})_j < c_j. 
\end{equation}

The second condition is that, for every agent $i$ and a pair of deterministic bundles $\y'\succ_i \y$, where $\y$ is assigned with positive probability in $\lottery{x}{i}$, we have 
\begin{equation} \label{eq:c2}
    \prices\cdot (\y'-\y) \ge 1.
\end{equation}

The second condition holds because if an agent consumes a bundle with a positive probability, it implies that there exists a realization of the budget such that it is the optimal bundle within that budget constraint. Consequently, the cost of any more preferred bundle must be strictly higher. By scaling prices, we can assume that it is at least 1.

Conversely, if there exists a set of prices, denoted as $\prices$, satisfying Eqs. \eqref{eq:c1} and \eqref{eq:c2}, then $((\lottery{x}{1},\ldots,\lottery{x}{n}),\prices)$ forms a CERI under the budget distribution $\randombudget{i}:= \prices \cdot \x$ with a probability distribution of $\mathbb{P}(\lottery{x}{i}=\x)$ for each agent $i$.

Thus, when provided with a feasible allocation $\lottery{X}{} =(\lottery{x}{1},\ldots,\lottery{x}{n})$, we can determine whether it is a CERI by using a linear program that identifies prices $\prices$ satisfying both Eqs. \eqref{eq:c1} and \eqref{eq:c2}.

Let $\{\mathcal{A}\prices \ge \bf{1}, \prices \ge 0\}$ be the description of this linear program. By   Farkas' lemma, such $\prices$ exists if and only if $\{\mathcal{A}^T \lambda \le 0,  \lambda \ge 0; {\bf 1}^T \lambda >0\}$ does not have a solution.

We will show that  the dual has a solution if and only if  $\lottery{X}{} =(\lottery{x}{1},\ldots,\lottery{x}{n})$ is not ordinally efficient. 
To show it, we first observe that an equivalent expression for stochastic dominance  $\lottery{y}{} \succeq^{sd}_i \lottery{x}{}$
is that if $\lottery{y}{}$ can be derived from $\lottery{x}{}$  through a sequence of moves involving the reallocation of probabilities from less preferred to more preferred bundles.

The dual variable $\lambda$ has a clear interpretation. Each coordinate corresponds to a tuple $i,\y',\y$, where $\y'\succ \y$ and 
$\mathbb{P}(\lottery{x}{i}=\y)>0$.   $\mathcal{A}^T \lambda \le 0$ implies that   for every good $j$  that is fully allocated,
\begin{equation}\label{eq:dual}
    \sum_{i,\y',\y} \lambda_{i,\y',\y}\cdot y'_j  -  \sum_{i,\y',\y} \lambda_{i,\y',\y} \cdot y_j \le 0
\end{equation}

Due to the positive probability $\mathbb{P}(\lottery{x}{i}=\y)$, we can normalize $\lambda$ such that each $\lambda_{i,\y',\y}$ is at most $\mathbb{P}(\lottery{x}{i}=\y)$.

Interpret $\lambda_{i,\y',\y}$ as the probability mass indicating the extent to which agent $i$ shifts from bundle $\y$ to a superior bundle $\y'$. Eq.~\eqref{eq:dual} asserts that agents can make such shifts without violating the resource constraint of goods that are fully allocated.

For goods that are under-allocated in the allocation $\lottery{X}{}$, Eq.~\eqref{eq:dual} may not hold. However, it is possible to additionally scale down $\lambda$ by a constant factor, ensuring that the difference is within the limits of the remaining resources for that particular good:

\begin{equation} \label{eq:dual1}
\sum_{i,\y',\y} \lambda_{i,\y',\y}\cdot y'_j  -  \sum_{i,\y',\y} \lambda_{i,\y',\y} \cdot y_j \le c_j- (\sum_i \expect{\lottery{x}{i}})_j.
\end{equation}

Let $\lottery{z}{i}$ represent the resulting lottery for agent $i$ after the probability shift according to $\lambda$. Because all agents have shifted probability to more preferred bundles,  $\lottery{z}{i}\succeq_i^{sd}\lottery{x}{i}$. Moreover, ${\bf 1}^T \lambda >0$ indicates that there is at least one move of probability that is strictly positive, thereby implying that at least one agent is strictly better off.

Eqs.~\eqref{eq:dual} and \eqref{eq:dual1} together imply that the allocation $({\lottery{z}{1}},\ldots,{\lottery{z}{n}})$ is feasible with respect to $\supplyvector$. This means that the allocation $\lottery{X}{}$ is not ordinally efficient.
      \qed

\subsection{Proof of Theorem~\ref{theo:envy}}
\paragraph{Part (i)} Let $\mathbf{p}^*$ denote the equilibrium price. Consider the random bundles $(\lottery{x}{i}, \lottery{x}{i'}$) consumed by individuals $i$ and $i'$ when they share a common income distribution $\mathcal{B}$. Let $u_i$ represent a cardinal valuation consistent with the preference relation $\succ_i$.

For a given budget $b$, let $\x_{(\mathbf{p}, b, \succ_i)}$ denote  the best affordable bundle at price $\mathbf{p}$ under the preference relation $\succ_i$. Clearly,
$
\x_{(\mathbf{p}^*, b, \succ_i)} \succeq_i \x_{(\mathbf{p}^*, b, \succ_{i'})}.
$
Since the budget distributions $\mathcal{B}$ are identical, it must hold that for any bundle $z \in \Psi_i$
\[
\sum_{\x \succeq_i \z} \mathbb{P}_{\x}(\lottery{x}{i}) \;\;\ge\;\; 
\sum_{\y \succeq_i \z} \mathbb{P}_{\x}(\lottery{x}{i'}).
\]
Hence,  $\lottery{x}{i} \succeq_i^{sd} \lottery{x}{i'}$ as required.
\paragraph{Part (ii)}
We omit the proof of the second part of the theorem as it  follows directly from Theorem 3 in~\citet{budish2011combinatorial}. \qed
 \subsection{Proof of Theorem~\ref{thm:m1}}\label{proof:thm1}

 We  need an additional definition for the proofs.  Let $\typedist$ be a discrete type distribution. Define $\tau = | \support (\typedist) |$.  For every $t \in \support(\typedist)$, denote by  $p_t$ the probability that $t$ is sampled from $\typedist$; $p_t = \prob_{\succ \sim \typedist}(\succ = t)$, and set $\minp = \min_{t \in \support(\typedist)}\{p_t\}$. 
For simplicity, we assume that $\tau \sqrt{n}$ is an integer. 

The following three lemmata will be used to derive the efficiency and incentive properties of \Mone. Lemma~\ref{m1:2} and Lemma~\ref{lem:usingchernoff} derive statistical properties of the random grid. 
Lemma~\ref{lem:typediscrepancy} shows that an allocation that is efficient with respect to one type vector is almost  efficient with respect to another type vector, assuming they are sufficiently close.

\begin{lemma} \label{m1:2} Let $\nagtypevec = (\nagtype^1, \nagtype^2, \ldots, \nagtype^{\numtypes })$ be a vector of $\numtypes$ types, and let $\grid{\gridjump = \numtypes \sqrt{n},\numtypes }$ be a random grid. Let $\epsilon>0$, and assume that $n \geq \nthresh$, where $\nthresh =  \frac{36}{\epsilon^2}$.  Then with probability at least $1-\frac{\epsilon}{2}$,   for every $\mathbf{u} = (u^1, \ldots, u^\numtypes) \in \mathbf{G}$, it holds that for every $t \in [\numtypes]$, $|\nagtype^t - u^t|>1$.
\end{lemma}
\begin{proof}
    For every $t \in [\numtypes]$, let $u^t$ denote the coordinate of $\mathbf{G}^t$ that is closest to $\nagtype^t$. As the grid is chosen at random, the probability that  for any $t \in [\numtypes]$, 
    \begin{equation}
        \prob\left(|\nagtype^t - u^t|\leq 1\right) = \frac{3}{\numtypes\sqrt{n}} \leq \frac{3}{\numtypes\sqrt{\nthresh }}  = \frac{\epsilon}{2\numtypes},
    \end{equation}
    as there are three possible values of $u^t$ for which $|\nagtype^t - u^t|\leq 1$. 
    Taking a union bound over the types completes the proof. 
\end{proof}

\begin{lemma}\label{lem:usingchernoff}
    Let $\epsilon>0$, and let $\typedist$ be a discrete type distribution. Let $\grid{\gridjump,\numtypes }$ be a grid  with  $\numtypes = |\support(\typedist)|$.   Assume that $n \geq \nthresh$, where $\nthresh =  \frac{(8\numtypes )^2}{\epsilon^2\minp^2}$. Let  $\mathcal{E} = (\nagtypevec, \capacities)$ be  a random economy where $\nagtypevec~\sim \typedist^n$. Let  $\typeapprox$ be the smallest grid point that is at least as large as $\nagtypevec$, coordinate-wise. Then with probability at least $1-\frac{\epsilon}{2}$,  $ \nagtypevec \geq \left(1-\frac{\gridjump}{0.5\minp n}\right) \typeapprox$.
\end{lemma}
\begin{proof}
 
Recall that $\nagtype^t$ denotes the number of agents of type $t$; let $\mu_t$ denote the expectation of $\nagtype^t$; that is $\mu_t = \expectation_{\nagtypevec\sim \typedist^n} (\nagtype^t) = np_t$.

Using a standard Chernoff bound, we have that for any $\delta \in (0,1)$, 
$\prob\left(\nagtype^t  \leq (1-\delta)\mu_t\right)\leq e^{-\delta^2\mu_t/2}. $
Setting $\delta = \sqrt{\frac{2\log{n}}{\minp n}}$ , let $\eventb_t$ denote that event $\nagtype^t \geq (1-\delta)\mu_t$, and let $\eventb$ denote the event that $\eventb_t$ holds for all $t \in \support(\typedist)$. 
We have  that for any $t \in \support(\typedist)$, 
\begin{equation}\prob\left(\neg \eventb_t \right) \leq e^{-\log{n}} = \frac{1}{n}. 
\end{equation}
Therefore, using the union bound, $\prob\left(\neg \eventb \right) \leq \frac{\numtypes }{n}<\frac{\epsilon}{2}$. Therefore $\eventb$ occurs  with probability at least $1 -\frac{\epsilon}{2}$. Note that as $n \geq \frac{64}{\minp^2}$, 
\begin{align*}
\delta = \sqrt{\frac{2\log{n}}{\minp n}} \leq  \sqrt{\frac{2}{\minp \sqrt{n}}} \leq \sqrt{\frac{2}{8}} = \frac{1}{2}.
\end{align*}
The first inequality is because $\frac{\log{n}}{n} \leq \frac{1}{\sqrt{n}}$.
Therefore as long as  $\eventb$ occurs then $t \in \support(\typedist)$, $\nagtype^t \geq 0.5\minp n$. We bound the ratio of $\typeapprox$ and $\nagtypevec$, assuming that $\eventb$ indeed occurs:

\begin{align*}
    \frac{\nagtypevec}{\typeapprox} \geq \frac{\nagtypevec}{\nagtypevec + \gridjump} %
    = 1-\frac{\gridjump}{0.5\minp n+\gridjump} \geq 1-\frac{\gridjump}{0.5\minp n}. %
\end{align*}

\end{proof}

\begin{lemma}\label{lem:typediscrepancy}
    Let $\epsilon>0$, and let $\nagtypevec \in \mathbb{R}^\numtypes $ be the the vector of the number of agents of each type. Let $\typeapprox \in \mathbb{R}^\numtypes $ be a vector such that $(1-\epsilon) \typeapprox \leq \nagtypevec \leq  \typeapprox$ and let $\lottery{X}{}$ be an allocation that is efficient with respect to the economy $\mathcal{E}' = \zeconomy$. Then $\lottery{X}{}$ is $\epsilon$-efficient with respect to $\mathcal{E} = \pseconomy$.  
\end{lemma}

\begin{proof}
    Let $x^t_j$ denote the (expected) amount of good $j$ allocated to an agent of type $t$ in the allocation $\lottery{X}{}$. As $\lottery{X}{}$ is efficient with respect to $\mathcal{E}' = \zeconomy$, by Theorem~\ref{thm:characterization}, there must be prices $\prices$ 
    such that 
    $$\sum_{t \in \numtypes } \typeapproxq^t x^t_j \leq c_j, $$
    with equality for each good $\good$ such that for $\price{\good}>0$.
    The (actual) expected {according to true types} amount of good $j$ allocated is $\sum_{t \in \numtypes } \nagtype^t x^t_j$. Then %
    for every $\good$ such that $\price{\good}>0$, \begin{align*}
      \sum_{t \in \numtypes } \nagtype^t x^t_j &\geq \sum_{t \in \numtypes } (1-\epsilon)\typeapproxq^t x^t_j
         =(1-\epsilon) \sum_{t \in \numtypes } \typeapproxq^t x^t_j = (1-\epsilon) c_j.
    \end{align*}
    Similarly, $\sum_{t \in \numtypes } \nagtype^t x^t_j \leq  c_j$. 
    For every $\good$ such that  $\price{\good}=0$, $\sum_{t \in \numtypes } \nagtype^t x^t_j < c_j$. Therefore, $\lottery{X}{}$ is $\epsilon$-efficient with respect to $\mathcal{E}$.
\end{proof}

To complete the proof of Theorem~\ref{thm:m1},
for any $\epsilon>0$, set $\nthresh = \frac{(8\numtypes )^2}{\epsilon^2\minp^2}$. 
\paragraph{Part (i): uniform strategyproofness}:
From Lemma~\ref{m1:2}, the coordinate-wise distance between the vector $\typeapprox$ generated by Mechanism \Mone\  and $\nagtypevec$ is   at least $1$, with probability at least $1-\epsilon$.  If this is the case, no agent can affect the choice of the grid point used for the generation of the prices. 
\paragraph{Part (ii): asymptotic ordinal efficiency}
Plugging  $\gridjump = \gridjumplong$ into Lemma~\ref{lem:usingchernoff}, gives that   $ \nagtypevec \geq \left(1-\epsilon\right) \typeapprox$ with probability at least $1-\frac{\epsilon}{2}$. Lemma~\ref{lem:typediscrepancy} shows that when this is the case, the allocation is asymptotically ordinally efficient.    Note that Lemma~\ref{lem:usingchernoff} does not require that the grid be random. As the only source of randomness in Mechanism \Mone\ is the choice of the random grid, we have that the allocation is asymptotically efficient for every realization of the mechanism; hence it is  asymptotically efficient. 
\paragraph{Part (iii): asymptotic ex-post efficiency} This property follows directly from the asymptotic ordinal efficiency of the mechanism and from applying Theorem~\ref{theo:feasibility} to \Mone\ allocation. Specifically, \Mone\ implements a CERI as a lottery that deviates from the expected aggregate allocation of the CERI by at most $\Delta-1$ units for each good.
\paragraph{Parts (iv) and (v): Ordinal envy-freeness and ex-post EF1}
Since the allocation is a CERI with respect to the economy $\mathcal{E}' = (\typeapprox, \mathbf{c})$, then using Theorem~\ref{theo:envy} we can conclude that it is {ordinally }envy-free and ex-post EF1.

\qed

\newpage
\section{CERI vs. ACEEI}
\subsection{Existence of a CERI implies existence of an \aceei}\label{app:B}
We first show that 
the existence of a CERI implies the existence of an approximate competitive equilibrium from equal incomes (\aceei) with the same bound on the {excess} demand as in \citet{budish2011combinatorial}. In particular we show the following result.
\begin{proposition}~\label{prop:budish}
     For any CERI, there exists an ex-post implementation which has the excess demand bounded by  $\sqrt{\Delta m/2}$ in the $\ell^2$-norm.
 \end{proposition}
To show this, we use the following improved bound for the Shapley--Folkman theorem \citep{budish2020improved} in order to obtain an  \aceei. 
\begin{lemma}[Theorem 3.1 in \citealp{budish2020improved}] \label{theo:bush} 
If $S_1,\ldots,S_n$ are compact subsets of $\mathbf{R}^m$, if ${\bf c}\in conv(S_1+..+S_n)$, and $D$ is the maximum diameter of $S_i$, then there exists $\x_i\in S_i$ such that 
    $$
    ||{\bf c}-\sum_i \x_i||_{\ell_2} \le D\sqrt{m}/2.
    $$
\end{lemma}

\begin{proof}[Proof of Proposition~\ref{prop:budish}]
Let $\randombudget{i}$ be a uniform distribution between $[1,1+\epsilon]$ for all agent $i$. 
Let $\lottery{x}{1},..,\lottery{x}{n}$ and $\prices$ be a CERI.  Apply Lemma~\ref{theo:bush}, where $S_i$ is the convex hull of all realization of $\lottery{x}{i}$, and ${\bf c}$ is the capacity vector.  If all agents' acceptable bundles are of size at most $\Delta$, then the diameter of $S_i$ is at most $\sqrt{2\Delta}$. We obtain the existence of 
$\x_i\in S_i$ such that $ ||{\bf c}-\sum_i \x_i||_2 \le \sqrt{\Delta m/2}.$ For each $\x_i\in S_i $, there is a realization $b_i$ of $\randombudget{\agent}$ such that 
$\x_i$ is the optimal choice of agent $i$ under budget $b_i$. Hence, the allocation $(\x_1,..,\x_n)$ corresponds to an ACEEI allocation in which the budget of each agent is perturbed by at most $\epsilon$.
\end{proof}

\subsection{The ACEEI mechanism is not ordinally efficient}\label{app:accei}  
In Section~\ref{sec:relationship}, we argued that in the unit-demand setting the ACEEI mechanism coincides with the RSD. An example from \citet{bogomolnaia2001new} described in that section showed that RSD allocations might not be ordinally efficient. Here, we provide an example of the ordinal inefficiency of the ACEEI in multiunit-demand setting. In this example, the unique ACEEI allocation (for any budget perturbation) is one that allocates the empty bundle w.p. 1 to all agents. This allocation turns out to be ordinally inefficient.

\begin{example}\label{ex:01}

The economy consists of a single good with capacity of 20 units, and two identical agents. The only bundle that both agents prefer to the empty bundle has 100 units. In the convexified equilibrium, each agent is allocated a lottery over the empty bundle and the bundle containing 100 units.  

ACEEI is a deterministic allocation that has excess demand bounded in the $\ell^2$-norm. In particular, the bound given by \citet{budish2011combinatorial} is for the case when agents demand at most a single unit of any good and is expressed as $\sqrt{Dm}/2$. This is generalized to the multiunit case in \citet{budish2020improved} as $D\sqrt{m}/2$, where $D$ represents the diameter of the choice set, and $m$ is the number of goods.

In this example, the diameter of the agents' choice set is 100, indicating an equilibrium where the $\ell^2$-norm of excess demand is at most $100/{2}=50$. If at least one agent receives the entire 100-unit bundle, the excess demand is 80. Consequently, the only allocation achieving the desired bound is the allocation in which agents receive nothing, which is not ordinally efficient.

\end{example}

\newpage
\bibliographystyle{chicago}
\bibliography{ref}

\end{document}